%% file: sastry_2d_delaunay.tex
\documentclass[11pt,english]{article}
\usepackage[a4paper,margin=1in]{geometry}
\usepackage{subcaption}
\usepackage{amssymb}
\usepackage{amsmath}
\usepackage{amsthm}
\usepackage{amsfonts}
\usepackage{verbatim}
\usepackage{graphicx}
\usepackage[titletoc,page]{appendix}
\usepackage{tikz}
\usetikzlibrary{calc, backgrounds,arrows,patterns}

\theoremstyle{plain}
\newtheorem{theorem}{Theorem}[section]
\newtheorem{lemma}[theorem]{Lemma}

\theoremstyle{definition}

\theoremstyle{remark}

\usepackage{mathtools}
\DeclarePairedDelimiter\ceil{\lceil}{\rceil}
\DeclarePairedDelimiter\floor{\lfloor}{\rfloor}
\begin{document}

\title{\Large 
A 2D Advancing-Front Delaunay Mesh Refinement Algorithm
\thanks{The work of the author was supported in part by the NIH/NIGMS Center
for Integrative Biomedical Computing grant 2P41 RR0112553-12.  The author would
also like to thank Ms. Christine Pickett for proofreading a draft of
the paper and suggesting numerous changes. 
}}
\author{Shankar P. Sastry\\C3 IoT, Redwood City, CA 94063, U.S.A.\\{\tt sastry@sci.utah.edu}}

\date{}

\maketitle
\input{abstract}
\pagebreak
\input{intro}

\input{related}
\input{background}
\input{algorithm}
\input{analysis0}

\input{analysis1}

\input{analysis2}
\input{future}


\newpage
\begin{appendix}
\input{appendix}
\end{appendix}
\pagebreak
\bibliographystyle{abbrv}
\bibliography{myrefs}

\end{document}

%% file: abstract.tex
\begin{abstract}
I present a generalization of Chew's first algorithm 
for Delaunay mesh refinement. In his algorithm, Chew
splits the line segments of the input planar straight line 
graph (PSLG) into shorter subsegments whose lengths are nearly
identical.  The constrained Delaunay triangulation of the
subsegments is refined based on the length of the 
radii of the circumcircles of the
triangles.  This algorithm produces a {\em uniform} mesh, 
whose minimum angle can be at most $\pi/6$.  My algorithm 
generates both truly Delaunay and constrained Delaunay 
{\em size-optimal} meshes.  In my
algorithm, I split the line segments of the input PSLG 
such that their lengths are asymptotically proportional 
to the local feature size (LFS) by solving ordinary 
differential equations (ODEs) that map points from a 
closed 1D interval to points on the input line segments
in the PSLG. I then refine the Delaunay triangulation 
(truly or constrained) of the PSLG
by inserting off-center Steiner vertices of ``skinny'' 
triangles while prioritizing such triangles with shortest 
edges first. As in Chew's algorithm, I show that the Steiner
vertices do not encroach upon any subsegment of the PSLG.  
The off-center insertion algorithm places Steiner vertices
in an advancing front manner such that 
we obtain a size-optimal Delaunay mesh 
(truly or constrained)
if the desired minimum angle is less than
$\pi/6$.  In addition, even in the presence of a small
angle $\phi < \pi/2$ in the PSLG, the bound on the minimum 
angle ``across'' the small angle tends to 
$\arctan{((\sin{\phi})/(2-\cos(\phi))}$
as the PSLG is progressively refined.  
Also, the bound on the maximum angle across any small 
input angle tends to $\pi/2 + \phi/2$ as the PSLG is 
progressively refined.  
\end{abstract}

%% file: intro.tex
\section{Introduction}
Delaunay mesh refinement techniques are commonly used to
generate high-quality meshes in two or higher dimensions.  
These meshes
are usually used to solve partial differential equations (PDEs)
with the finite volume method (FVM) or the finite element 
method (FEM)~\cite{S02b}.  The FVM is typically used to solve fluid-flow
problems by defining a control volume surrounding each vertex
in a mesh and measuring the flux entering and exiting the control
volume. Delaunay meshes are extensively used in the 
FVM because it is easy to define a control volume
around a vertex using a Delaunay mesh and its corresponding 
Voronoi diagram.  

The FEM, on the other hand, may be used with any 
mesh. For instance, in order to solve isotropic elliptic 
PDEs, the FEM requires meshes whose elements are 
``regularly'' shaped, i.e., the length of their edges should
be nearly identical~\cite{BA76}. In 2D triangular meshes, this regularity
requirement translates to triangles having angles as close to
60 degrees as possible.  In 2D, given a set of vertices, 
Delaunay triangulation maximizes the minimum angle among all
possible triangulations over the set of vertices.  As this
property aligns with the goal of mesh generation,
Delaunay meshes are also used in the FEM to solve isotropic
elliptic PDEs.  In this paper, I focus only on 2D meshes. 

A challenge in generating Delaunay meshes is to make the 
meshes conform to both internal and external boundaries,
which can be accomplished using either truly Delaunay meshes 
or constrained Delaunay meshes~\cite{C87}. 
A truly Delaunay mesh does not permit any mesh vertex
inside the circumcircle of a triangular element. With 
truly Delaunay meshes, the Delaunay triangulation of 
mesh vertices recovers the boundaries automatically. 
In the context of solving PDEs, truly Delaunay meshes 
are necessary when the FVM is being used, and the mesh
might be rather large.  For other applications (such as the
FEM), a constrained Delaunay meshes might be 
sufficient. 

A constrained Delaunay mesh permits
mesh vertices inside the circumcenter of a triangular 
element if the mesh vertices are not visible (explained
in Fig.~\ref{fig:tdtcdt}) from any point inside the 
triangular element.  
This flexibility allows the generation of smaller 
high-quality meshes when compared to truly Delaunay meshes. 
Fig.~\ref{fig:tdtcdt} shows some examples of truly Delaunay
triangulation and constrained Delaunay triangulation.  
Note that it is easy to show that if no vertices 
are allowed inside the diametral circles of subsegments 
of an input segment, a Delaunay triangulation of vertices
will yield a conforming Delaunay mesh that is truly 
Delaunay. 

A typical Delaunay mesh refinement algorithm starts from an 
initial Delaunay triangulation of the input geometric
domain.  The input geometric domain is also called a
planar straight line graph (PSLG) since the input is planar
and consists of vertices and line segments that may be 
thought of as an embedded graph. 
The Delaunay refinement algorithms progressively add 
vertices to the mesh and retriangulate the vertices such 
that poorly shaped triangles are eliminated.  
In Section 2, I will discuss some of the techniques developed 
to obtain Delaunay meshes. Most algorithms
liberally add vertices (also called Steiner vertices) 
within the input domain.  When the techniques, however,  
attempt to add a Steiner vertex close to the boundaries 
or outside the domain, it is called an encroachment of 
the domain. Most algorithms 
handle this case by splitting the relevant segments in the 
input PSLG, deleting a few other vertices (only 
in some algorithms), and 
retriangulating the remaining vertices. 
These technique produce well-graded, high-quality meshes. 

My technique is a generalization of a technique by 
Chew~\cite{C89}.  Chew 
splits the input segments in the PSLG into 
subsegments whose lengths are nearly identical. 
The split PSLG segments are triangulated and
refined by adding Steiner vertices, which are
circumcenters of triangles whose radius of the
circumcircle is larger than the length of the
shortest subsegment in the split PSLG. 
The technique never attempts to
add a vertex outside the domain. 
As a result of the nearly uniform splitting of
the PSLG segments and the refinement technique, 
this algorithm produces uniform meshes.  

In my technique, I generate well-graded meshes by
refining the PSLG such that the lengths of the split
segments are asymptotically proportional to the 
local feature size (formally defined in Section 3)
at the end points of the split segments. Such
asymptotically proportional splits are 
achieved by solving an ordinary differential equation,
whose solution maps points from a reference 
line segment to a line segment in the PSLG. 
I then refine poor-quality triangles by adding their
circumcenter or an off-center Steiner vertex~\cite{EU09,U09,U04}. 
I prioritize poor-quality triangles with shortest edges first. 
More details are provided in Section 4.  In Section 5, 
I show that my adaptive splitting also ensures that 
the algorithm never attempts to insert a vertex 
outside the domain. I have separately analyzed
the algorithm for generating 
truly and constrained Delaunay meshes. 

My technique improves the upper bound on the minimum
angle in a size-optimal mesh to 30 degrees (except
``across'' small angles).  Moreover, even in the 
presence of small angles, the technique improves the 
lower bound on the maximum angle to $\pi/2 + \phi/2$,
where $\phi<\pi/2$ is the smallest angle in the 
input PSLG.  These results hold for both truly
and constrained Delaunay meshes, but the constants
associated with size optimality (defined in Section 3) 
are different for the two types of meshes.  As
expected, truly Delaunay meshes are larger than
constrained Delaunay meshes. 

\begin{figure}
\centering
\begin{subfigure}[b]{0.23\textwidth}
\begin{tikzpicture}[scale=2]
    \draw [line width=0.75mm] (0,0) -- (1,0);
    \draw [line width=0.15mm] (0,0) -- (0.25,0.5);
    \draw [line width=0.15mm] (0.5,0) -- (0.25,0.5);
    \draw [line width=0.15mm] (0.5,0) -- (0.75,0.5);
    \draw [line width=0.15mm] (0.25,0.5) -- (0.75,0.5);
    \draw [line width=0.15mm] (1,0) -- (0.75,0.5);
    \draw [line width=0.15mm] (0,0) -- (0.25,-0.5);
    \draw [line width=0.15mm] (0.5,0) -- (0.25,-0.5);
    \draw [line width=0.15mm] (0.5,0) -- (0.75,-0.5);
    \draw [line width=0.15mm] (1,0) -- (0.75,-0.5);
    \draw [line width=0.15mm] (0.25,-0.5) -- (0.75,-0.5);
    \draw [dotted] (0.25, -0.1875) circle (0.3125);
    \draw [dotted] (0.25,  0.1875) circle(0.3125);
    \draw [dotted] (0.5, -0.1875) circle (0.3125);
    \draw [dotted] (0.5,  0.1875) circle (0.3125);
    \draw [dotted] (0.75, -0.1875) circle (0.3125);
    \draw [dotted] (0.75,  0.1875) circle (0.3125);
\end{tikzpicture}
\caption{Truly\\Delaunay}
\end{subfigure}
\centering
\begin{subfigure}[b]{0.23\textwidth}
\begin{tikzpicture}[scale=2]
    \draw [line width=0.75mm] (0,0) -- (1,0);
    \draw [line width=0.15mm] (0,0) -- (0.5,0.3);
    \draw [line width=0.15mm] (1,0) -- (0.5,0.3);
    \draw [line width=0.15mm] (0,0) -- (0.5,-0.3);
    \draw [line width=0.15mm] (1,0) -- (0.5,-0.3);
    \draw [dotted] (0.5, -0.2667) circle (0.5667);
    \draw [dotted] (0.5,  0.2667) circle (0.5667);
\end{tikzpicture}
\caption{Constrained\\Delaunay}
\end{subfigure}
\centering
\begin{subfigure}[b]{0.23\textwidth}
\begin{tikzpicture}[scale=2]
    \draw [line width=0.75mm] (0,0) -- (1,0);
    \draw [line width=0.15mm] (0.25,0.1) -- (0.75,0.1);
    \draw [line width=0.15mm] (0.25,0.1) -- (0.5,0.2);
    \draw [line width=0.15mm] (0.75,0.1) -- (0.5,0.2);
    \draw [line width=0.15mm] (0,0) -- (0.5,-0.3);
    \draw [line width=0.15mm] (1,0) -- (0.5,-0.3);
    \draw [dotted] (0.5,  0.2667) circle (0.5667);
    \draw [dotted] (0.5, -0.1625) circle (0.3625);
\end{tikzpicture}
\caption{Constrained\\Delaunay}
\end{subfigure}
\centering
\begin{subfigure}[b]{0.23\textwidth}
\begin{tikzpicture}[scale=2]
    \useasboundingbox (-0.1,-0.5) rectangle (1,0.5);
    \draw [line width=0.75mm] (0,0) -- (1,0);
    \draw [line width=0.15mm] (-0.1,0.1) -- (0.35,0.1);
    \draw [line width=0.15mm] (-0.1,0.1) -- (0.05,0.2);
    \draw [line width=0.15mm] (0.35,0.1) -- (0.05,0.2);
    \draw [dotted] (0.125, -0.075) circle (0.285);
\end{tikzpicture}
\caption{Not\\allowed}
\end{subfigure}
\caption{
The difference between truly Delaunay and constrained Delaunay
meshes. Only parts of the meshes are shown for clarity. The thick 
horizontal lines are input segments to which the mesh should conform.  
A vertex $a$ is not visible from a point $b$ if the line
segment joining $a$ and $b$ intersects an input line segment. 
(a) All triangles are 
Delaunay because no vertex is inside their circumcircles. 
(b) A vertex from the other triangle is inside the 
circumcircle of both triangles, but the vertex is not visible from 
the interior of the triangles. 
(c) Similar to (b), but one of the triangles is not on the input
segment.
(d) This case is not allowed because an end point of the input
segment is inside the circumcircle of a triangle, and the
end point is visible from the interior of the triangle. 
}
\label{fig:tdtcdt}
\end{figure}
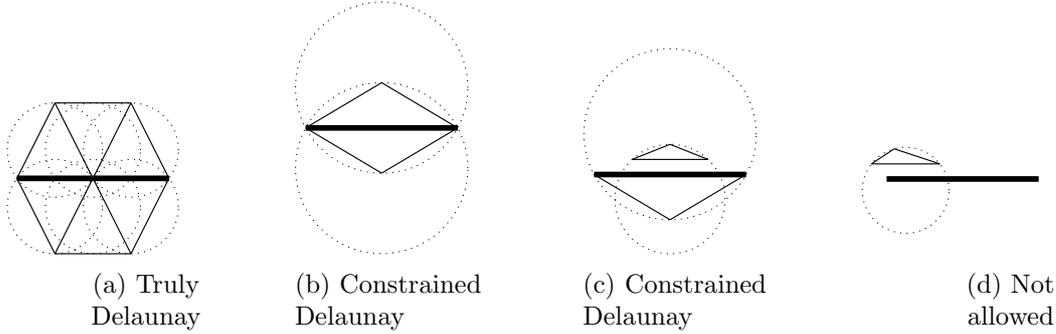

%% file: related.tex
\section{Related Work}
An extensive body of literature has focused on generating
2D Delaunay meshes with a specified minimum angle.  
The first subsection below reviews prior 
algorithms in which the PSLG does 
not contain any angle smaller than $\pi/2$ (or 
$\pi/3$, depending on the algorithm).  The second 
subsection reviews techniques
to deal with PSLGs with small angles and the 
limitations of any algorithm used to mesh
such PSLGs. 

In most algorithms, a triangle is considered to be of
poor quality or ``skinny'' if its minimum angle
$\theta$ is less than some user-defined threshold; it
is an input to the algorithm.  The ratio of the 
length of the radius of the circumcircle of a triangle
and its shortest length is equal to $1/(2\sin{(\theta)})$,
where $\theta$ is the minimum angle in the triangle
(see Fig.~\ref{fig:angleratio}). 
A user may equivalently provide the radius-edge ratio
to define a skinny triangle.   

\subsection{Delaunay Mesh Refinement}
The first Delaunay mesh refinement algorithm was
developed by Chew~\cite{C89} to obtain constrained
Delaunay meshes.  Chew
first splits the segments of the input PSLG such that
the length of each subsegment is between $h$ and 
$\sqrt{3}h$, where $h$ is small enough that such a 
division in possible\footnote{In his 
report~\cite{C89}, Chew provides more details about
how to find $h$ and how to split the PSLG.}.  
The split PSLG is then triangulated using the
constrained Delaunay triangulation algorithm~\cite{C87}.  
Then triangles whose radius of the circumcircle
is less than $h$ are ``split'' by adding its
circumcenter.  Upon retriangulation, the addition 
of the circumcenter eliminates the skinny triangle 
because the new vertex (the circumcenter) is inside
the skinny triangle, which Delaunay triangulation does 
not allow. Since Chew always adds vertices at a
distance of at least $h$ from other vertices 
(Delaunay triangulation ensures that there are no 
vertices inside the circumcircles of all triangles), 
we eventually run out of space to add more points.  In 
addition, Chew demonstrated that no triangle will be formed
such that the radius of its circumcircle is
less than $h$ and its circumcircle is outside the domain. 
Also, since the radius of the circumcircle is less than
$h$ and the shortest edge length is at least $h$, all the angles
in the mesh are greater than or equal to $\pi/6$. 
This algorithm results in a uniform mesh. In this paper,
I generalize this algorithm by developing a technique
to split PSLG segments in a size-optimal fashion rather
than uniformly.  

Ruppert~\cite{R93,R95} developed a similar Delaunay 
refinement algorithm that splits the PSLG on the fly
to generate truly Delaunay meshes.  Instead
of splitting a triangle based on the length of the
radius of the circumcircle, the radius-edge ratio became
the criterion.  A triangle is split only when the ratio is 
greater than 1 (it should be greater than $\sqrt{2}$ for reasons
explained below), which corresponds to the minimum angle being 
$\pi/6$.  Clearly, if the triangle is split when the ratio 
is less than 
$1$, the algorithm does not terminate as it places vertices
progressively closer to each other. If the Steiner vertex, however,
is outside the domain or inside the diametral circle of a PSLG
subsegment, the PSLG subsegment is split at its midpoint.  Because
the midpoint may introduce short edges, the threshold of the 
radius-edge ratio to split a triangle has to be increased to 
$\sqrt{2}$, which corresponds to a minimum  angle of about 
20.7 degrees. Rand~\cite{R10} showed that Ruppert's algorithm 
(for truly Delaunay meshes) terminates for angles almost 
as large as 22.2 degrees.
Ruppert showed that this algorithm provides a truly Delaunay,
size-optimal mesh (elements' edge lengths are graded based 
on their proximity to small features in the input PSLG) 
when the input angles in the PSLG are greater
than $\pi/2$. This threshold was later lowered to $\pi/3$
by Shewchuk~\cite{S97,S02}.  Shewchuk was also able to
improve the bound on the minimum angle in a triangle
to $\pi/6$ at the cost of size optimality.  He split
the PSLG segments such that their lengths were in
specific ranges, and this restriction in the subsegment
length resulted in the loss of the size-optimality guarantee. 

Chew~\cite{C93} independently devised a second constrained 
Delaunay refinement algorithm in which he incorporated the idea
of refining the PSLG only when necessary. In addition, 
the technique also removes points within the diametral 
circle of a subsegment when the subsegment is to be refined. 
Shewchuk~\cite{S97,S02} showed that this algorithm produces
size-optimal constrained Delaunay meshes when the desired
radius edge ratio is greater than $\sqrt{5}/2$, which 
translates to a minimum angle of about 26.57 
degrees.

Shewchuk~\cite{S97} improved upon both the techniques of
Chew and Ruppert above by refining triangles near any 
PSLG segment differently from those away from the 
segments.  Shewchuk's technique ensured that
the quality of triangle in the interior is better than
the ones near the PSLG boundary segments. Also, he 
introduced the notion of ``diametral lenses'' to analyze
Chew's second algorithm and showed the bound of
about 26.57 degrees mentioned above.  

Miller, Pav, and Walkington~\cite{MPW05, P03} showed that
in a modified version of Ruppert's technique (for truly 
Delaunay triangulation), at least three circumcenters 
have to be inserted between two refinements of a 
PSLG segment/subsegment.  Thus, in the
worst case and when the input is restricted (conditions
on the angles between PSLG segments and their lengths), 
they were able to show that the algorithm
terminates with a truly Delaunay, size-optimal mesh if
the minimum desired angle is set to some value strictly 
less than $\arcsin{2^{(-7/6)}} \approx 26.45$ degrees 
even when the input angle is as small
as $\pi/4$. Rand~\cite{R11} extended this analysis to
Chew's second algorithm (for constrained Delaunay triangulation) 
and showed that it can produce a size-optimal
mesh with a minimum desired angle is set to some value strictly 
less than $\theta$ such that $8 \sin^3{\theta}/\cos{\theta} < 1$,
which corresponds to about 28.60 degrees. 

\"{U}ng\"{o}r and Erten~\cite{EU09,U09,U04} developed a
heuristic technique to place Steiner vertices at 
``off-center'' points
such that the shortest edge of skinny triangles
subtend the desired minimum angle at the points.  
When coupled with the prioritization of 
skinny triangles with shortest edges first, they found that
the resulting meshes have fewer elements and vertices.  
Their technique works for both truly and constrained 
Delaunay meshes. I use the heuristic technique in my
algorithm, and it plays an important role in the proofs. 

Foteinos et al.~\cite{CC09, FCC10} generalized Chew's
second algorithm to show that the insertion of Steiner 
vertices can happen at any point (not just the 
circumcenter or the off-center point) in a skinny triangle's 
``selection circle'', which is a circle concentric with
its circumcircle, but with a radius shorter by the length
of the shortest edge of the skinny triangle. In addition, 
they showed that
the PSLG may be split at any point (not just the midpoint)
sufficiently far away from the end points.  Such an algorithm 
would still produce a size-optimal constrained Delaunay 
mesh with quality guarantees.  Hudson~\cite{H08} generalized
the selection circle to a larger selection region. 

\subsection{Managing Small Input Angles}
If the PSLG has any small input angles, the results
above do not hold for the whole mesh.  One may assume that
if the algorithms above ignore triangles at small input 
angles, it should be possible to construct a mesh with
larger angles everywhere else. Unfortunately, 
Shewchuk~\cite {S97,S02} showed that it is impossible to 
construct such a mesh. Researcher have since attempted
to mitigate the effects of small input angles. 

Shewchuk~\cite{S97,S00,S02} developed a technique
he called Terminator that splits skinny triangles
near a small input angle only if the new edge that
would be introduced in the mesh were long enough to 
ensure that the algorithm terminates. This technique 
(for truly Delaunay meshes) does
not have a theoretical guarantee of size optimality.  
He showed that the minimum angle in the mesh may be as large as
$\arcsin{\left(\frac{1}{\sqrt{2}}\sin{\left(\frac{\phi}{2}\right)}\right)}$,
where $\phi$ is the smallest angle in the input PSLG. 
He applied a similar technique for Chew's algorithm
(for constrained Delaunay meshes) and improved
the bounds to 
$\arcsin{\left(\frac{\sqrt{3}}{2}\sin{\left(\frac{\phi}{2}\right)}\right)}$. 

Rand and Walkington~\cite{RW09} employed a
collar-based strategy to protect regions around
a PSLG vertex with small angles and prevented
vertex insertion near a small angle. Pav and
Walkington~\cite{PW05} employed a similar
strategy. This strategy was inspired by
Ruppert's heuristic in his papers~\cite{R93,R95}.

Pav et al.~\cite{MPW05,P03} redefined a poor-quality triangle
as a skinny triangle whose end points of the 
shortest edge do not lie on two segments emanating 
from an input vertex with a small angle. 
They split all skinny triangles in the mesh that
were not ``across'' a small angle.
They ignored skinny triangles across a small input angle,
but due to the segment splits beforehand, they showed
that the angles in such skinny triangles were 
at least $\arctan{\left(\frac{sin{\phi}}{2-\cos{\phi}}\right)}$.  
This technique was inspired by Ruppert's concentric
shell splitting heuristic. 
In my algorithm, I split the PSLG before splitting
triangles.  Thus, I am also able to show almost
identical bounds.  In addition, the Pav et al. algorithm
produces a mesh with a maximum angle of 
$\pi - 2\arcsin{\frac{\sqrt{3}-1}{2}} \approx 137$ degrees.  
My algorithm improves this bound. 

%% file: background.tex
\section{Background}
The notations, concepts and algorithms presented in this 
section are used in the advancing-front algorithm.  

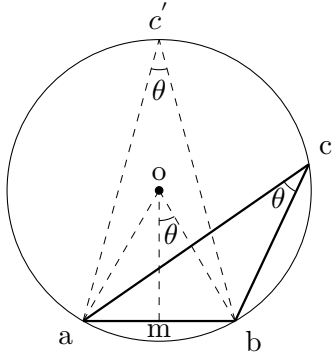
\begin{figure}
\centering
\begin{tikzpicture}[scale=2]
    \draw (0,0) circle (1);
    \draw [line width=0.30mm] (240:1) -- (300:1);
    \draw [line width=0.30mm] (300:1) -- (10:1);
    \draw [line width=0.30mm] (240:1) -- (10:1);
    \draw [line width=0.10mm, dashed] (300:1) -- (90:1);
    \draw [line width=0.10mm, dashed] (240:1) -- (90:1);
    \draw [fill=black] (0,0) circle (0.025);
    \draw [line width=0.10mm, dashed] (300:1) -- (0:0);
    \draw [line width=0.10mm, dashed] (240:1) -- (0:0);
    \draw [line width=0.10mm, dashed] (270:0.866025) -- (0:0);

    \draw  (90:1) ++(255:0.2) arc (255:285:0.2);
    \draw  (10:1) ++(213:0.2) arc (213:243:0.2);
    \draw  (0:0) ++(270:0.2) arc (270:300:0.2);

    \node [above] at (0,0) {o};
    \node [below left] at (240:1) {a};
    \node [below right] at (300:1) {b};
    \node [below] at (270:0.82){m};
    \node [above right] at (10:1) {c};
    \node [above] at (90:1) {$c^{'}$};
    \node [below] at (90:0.8) {$\theta$};
    \node at ($(10:1)+(228:0.3)$) {$\theta$};
    \node at (285:0.3) {$\theta$};
\end{tikzpicture}
\caption{
The relationship between the minimum angle in a triangle
and the ratio of the radius of its circumcircle and the
length of the shortest edge. Let $ab$ be the shortest
edge of $\triangle abc$.  Then, the angle at $c$ is its 
shortest angle.  The triangle's circumcircle is shown. 
Note that $\angle c^{'} = \angle c$. Also, $\angle aob = 
2\angle c^{'}$, and $\angle mob = \angle c^{'}$.  Clearly,
$\sin{\angle c} = \sin{\angle c^{'}} = 
\sin{\angle mob} = \frac{|ab|}{2|ob|}
= \frac{l}{2r}$, where $l = |ab|$ and $r$ is the length
of the radius of the circumcircle. 
}
\label{fig:angleratio}
\end{figure}

\subsection{A ``Skinny'' Triangle}
In this paper, I will use $\theta^*$ to denote the desired
minimum angle in a mesh and  $\alpha = 1/(2\sin{(\theta^*)})$ 
to  denote the desired radius-edge ratio (see 
Fig.~\ref{fig:angleratio} for an explanation).  A skinny triangle
is any triangle whose minimum angle is less than $\theta^*$
or whose radius-edge ratio is greater than $\alpha$. 

\subsection{The Local Feature Size and Size Optimality}
The local feature size, denoted as $\mathrm{LFS}(x)$ or simply
$F(x)$, is the
radius of the smallest circle centered at a point $x$ that 
intersects two nonadjacent features of the input PSLG.  Note 
that the local feature size is dependent only on the input but 
not on the mesh generated.  For any point $x$ on a line segment 
$pq$ in the input PSLG, since $x$ is already on an input feature,
the feature size at $x$ is the minimum of (a) the distance 
to the nearest feature from $x$ such that
the feature is not adjacent to $pq$ and 
(b) the distance to 
$p$ or $q$, whichever is farthest from $x$.  The distance
to $p$ or $q$  is also considered 
because $p$ and $q$ are not adjacent features, and both
$p$ and $q$ are inside a disk centered at $x$ with radius 
equal to $xp$ or $xq$, whichever is greater.

Ruppert~\cite{R93} introduced this metric
and showed that his algorithm produces meshes in which the
length of the edges is greater than some fraction of the
local feature size at its end points. In other words, the
length of any edge in the final mesh produced by his 
algorithm is greater than
$\gamma\mathrm{LFS}(x)$, where $\gamma$ is some constant. 
He called the meshes size-optimal meshes. 
Since then, many researchers have used this metric 
to show that their algorithms also produce meshes that 
are size optimal.

The local feature size is a Lipschitz continuous function,
i.e., $\mathrm{LFS}(y) \le \mathrm{LFS}(x) + ||x-y||$ and
$\mathrm{LFS}(y) \ge \mathrm{LFS}(x) - ||x-y||$  because
the nonadjacent features that are contained in a disk 
centered at $x$ with a radius $\mathrm{LFS}(x)$
are also contained in the disk centered at 
$y$ with a radius $\mathrm{LFS}(x)+||x-y||$.  

\subsection{Nonhomogenous Ordinary Differential Equations}
My algorithm involves the solution of ordinary differential
equations (ODEs) of first and second order.  There is an 
extensive body of literature addressing the solution of 
such equations analytically 
and/or numerically.  Fortunately, it is possible to derive an
analytical solution to the equations presented in Section 4. 
Below, I will go through the two equations we will
see in Section 4. 

First, let us consider an equation of the form 
$y\mathrm{'} + ay = b$, where $a$ and $b$ are constants.  
Multiplying both sides by
$e^{ax}$, we get $e^{ax}y\mathrm{'} + aye^{ax} = be^{ax}$.  
Integrating both sides w.r.t. $x$, we get
$e^{ax}y = (b/a)e^{ax} + c$, where $c$ is some constant.  
Thus, $y = (b/a) + c/e^{ax}$, where $c$ is determined using
a boundary condition.  If $a=0$, then $y = bx + c$, where 
$c$ is some constant. 

Second, let us consider an equation of the form
$y\mathrm{''} = y + a$, where $a$ is a constant.  The
solution to this equation is given by 
$y = c_1e^x + c_2e^{-x} - a$,
where $c_1$ and $c_2$ are constants to be determined using 
two independent boundary conditions.  

I use the solution of the differential equations discussed
above to adaptively split all line segments of the input
PSLG such that their lengths are asymptotically 
proportional to the local feature size.  
Given a line segment $pq$ that is a subsegment of the PSLG
line segment, I ensure that the length $l$ of $pq$ is such
that $A^*\le\frac{\mathrm{LFS}(p)}{l}\le B^*$, where 
$A^*$ and $B^*$ are some constants. I will denote the 
ratio $B^*:A^*$ as $R$.
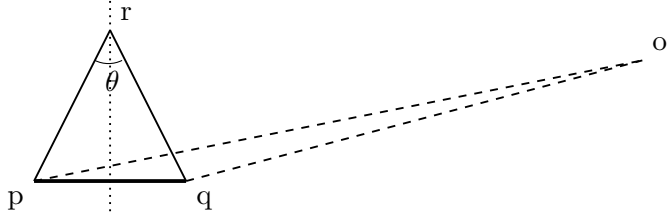
\begin{figure}
\centering
\begin{tikzpicture}[scale=2]
    \draw [line width=0.50mm] (0,0) -- (1,0);
    \draw [line width=0.25mm] (0,0) -- (0.5,1);
    \draw [line width=0.25mm] (1,0) -- (0.5,1);
    \draw [line width=0.25mm, dashed] (0,0) -- (4,0.8);
    \draw [line width=0.25mm, dashed] (1,0) -- (4,0.8);
    \draw [line width=0.25mm, dotted] (0.5,-0.2) -- (0.5,1.2);
    \draw (0.4,0.8) arc (243.44:296.56:0.2);
    \node [below left] at (0,0) {p};
    \node [below right] at (1,0) {q};
    \node [above right] at (0.5,1) {r};
    \node [above right] at (4,0.8) {o};
    \node [below right] at (0.4,0.8) {$\theta$};
\end{tikzpicture}
\caption{The off-center vertex. Instead of the circumcenter of 
a skinny triangle $pqo$, \"{U}ng\"{o}r and Erten~\cite{EU09} add an 
off-center vertex into the Delaunay triangulation.  The off-center
vertex $r$ lies on the perpendicular bisector of the shortest edge
$pq$ of the skinny triangle such that the edge subtends the minimum
desired angle at the point.  Note that $r$ should be on the same
side of $pq$ as $o$ is. As in their algorithm, I use the same
point for Delaunay refinement if it is closer to the shortest
edge than to the circumcenter.}
\label{fig:offcenter}
\end{figure}
\subsection{The Off-Center Refinement Algorithm}
In earlier Delaunay mesh refinement algorithms, Steiner
vertices were added at the circumcenter of skinny triangles 
in order to produce a mesh with no skinny triangles. 
Erten and \"{U}ng\"{o}r~\cite{U04,EU09} developed an 
algorithm in which the position of a Steiner vertex is an
off-center point on the perpendicular bisector of the shortest
edge of the skinny triangle.  The point is chosen such that
the angle subtended by the shortest edge at that point is
equal to the desired angle $\theta^*$ specified in the input
(see Fig.~\ref{fig:offcenter}).  If the off-center point is farther 
than the shortest edge than the circumcenter, the algorithm 
reverts to the circumcenter insertion technique. The 
algorithm processes poor-quality triangles with the 
shortest edges first, i.e., it considers the shortest
edge in every skinny triangle, and processes the triangle
with the shortest edge among those considered edges. 

In this algorithm, if the length of the shortest edge 
of a skinny triangle is $l$, the location of the off-center 
point is at a distance of $\beta l$ from the vertices of
the shortest edge, where $\beta = 1/(2\sin{(\theta^*/2)})$ and
$\theta^*$ is the 
desired minimum angle in the mesh.  Note that $\alpha <
\beta$, where $\alpha = 1/(2\sin{(\theta^*)})$ is the desired 
radius-edge ratio.  Since there are no other vertices
within the circumcircle of the skinny triangle (the Delaunay
property), this algorithm always places
Steiner vertices that are at least a distance of
$\alpha l$ from all other vertices in the mesh and at most
a distance of $\beta l$ from both vertices of the shortest
edge of the skinny triangle.  If $\alpha > 1$ (and if all the
Steiner vertices are placed inside the domain), it is easy to
see that the algorithm terminates when it runs out of place to
add more vertices in the domain. Also note that users could
purposefully insert a Steiner vertex at a distance between
$\alpha l$ and $\beta l$ if they choose to slowly grow 
the size of elements as the vertices are placed away from 
the PSLG segments. 

In my algorithm, I use the off-center Steiner 
vertex insertion algorithm described
above (including the prioritization of shortest edges). 
I ensure that the adaptive splitting results in Steiner vertex
insertion that does not encroach (defined below) 
upon any of the split segments.  

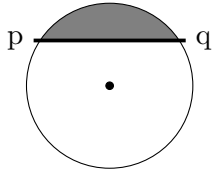
\begin{figure}
\centering
\begin{tikzpicture}[scale=2]
    \draw [line width=0.50mm] (0,0) -- (1,0);
    \node [left] at (0,0) {p};
    \node [right] at (1,0) {q};
    \begin{scope}[on background layer]
     \begin{pgfinterruptboundingbox}
       \clip (0.5,-0.3) circle (0.547);
       \clip (0,0) rectangle (1,1);;
     \end{pgfinterruptboundingbox}
       \fill [color = gray] (0,0) rectangle (1,1);
    \end{scope}
       \draw (0.5,-0.3) circle (0.547);
       \draw [fill=black] (0.5,-0.3) circle (0.025);
\end{tikzpicture}
\caption{Encroachment of a Steiner vertex. A possible
circumcircle of a skinny triangle is shown. Since the triangulation
obeys the Delaunay property, the circumcircle cannot enclose
any vertex of a PSLG subsegment $pq$. The circumcenter $o$ and the 
skinny triangle
have to be on opposite sides of $pq$.  Thus,
the skinny triangle has to be in the shaded gray region, which
is fully inside the diametral circle of the PSLG subsegment.}
\label{fig:insidedia}
\end{figure}

\subsection{Line Segment Encroachment}
All the algorithms mentioned in this paper work without 
issues only if the Steiner vertices are always placed inside
the domain.  When a potential vertex is outside the
domain, the vertex is considered to have encroached upon
the domain. More formally, a Steiner vertex encroaches upon 
a PSLG subsegment if the skinny triangle that resulted
in the Steiner vertex and the Steiner vertex lie on
the opposite sides of the PSLG line segment.  
Note that this definition considers the 
internal boundaries in the PSLG as well.  
Since the triangulation obeys the Delaunay property, 
the circumcircle of a skinny triangle cannot enclose
any vertex of a PSLG subsegment. For the 
circumcenter to encroach upon the subsegment, the 
circumcircle and the skinny triangle have to be
on different sides of the PSLG subsegment. This 
implies that the skinny triangle has to be inside 
the diametral circle of the PSLG subsegment 
(see Fig.~\ref{fig:insidedia}). 
In the algorithm, I insert the off-center vertex
or the circumcenter, whichever is closer. 
Therefore, I will always insert a vertex
that is at most the distance to the circumcenter. 
Thus, if the circumcenter does not encroach any PSLG
subsegment, we can be guaranteed to insert a
vertex that does not encroach upon a PSLG subsegment. 
If a Steiner vertex encroaches upon a PSLG line 
segment, prior algorithms impose 
corrective measures.  As
in Chew's algorithm, I will show that no circumcenter
of skinny triangles encroaches upon a PSLG line 
segment if the line segments are appropriately split.

\subsection{Small Angle}
I denote a small angle as $\phi$. 
Given a size-optimal splitting of the PSLG such that
$A^*\le\frac{\mathrm{LFS}(p)}{l}\le B^*$ for any subsegment
$pq$ of length $l$ and $R=B^*/A^*$, I define a small angle
as any angle $\phi\le\arccos{\frac{1}{2R}}$.

\subsection{Skinny Triangles ``Across'' a Small Angle}
Miller, Pav, and Walkington~\cite{MPW05,P03} developed
an algorithm that provides guarantees on the mesh quality
even in the presence of small angles in the PSLG.  In
their algorithm, they first adaptively split the line 
segments of the PSLG that form a small angle 
$\phi<\pi/3$ such that their lengths are
in the powers of two (in some global scale). 
During the mesh refinement phase, 
Miller et al. ignore skinny triangles if the 
vertices of their shortest edges lie on adjacent
line segments of the PSLG that form an angle 
$\phi<\pi/3$.  As in their algorithm, I too ignore 
skinny triangles ``across'' a small angle, but
I will define the small angle
as $\phi\le\arccos{\frac{1}{2R}}$ as
in subsection above. 

%% file: algorithm.tex
\section{The Advancing Front Algorithm}
The advancing front algorithm carries out the following 
three steps in succession to generate a size-optimal
mesh:
\begin{enumerate} 
\item The computation of the piecewise-smooth local feature size
functions for the input line segments of the PSLG.
\item The splitting of the input line segments into subsegments 
whose lengths are asymptotically proportional to the 
local feature size.
\item The refinement of the truly or the constrained 
Delaunay triangulation of the PSLG until all skinny triangles 
are eliminated. 
\end{enumerate}
Each step is described in 
detail in the following subsections. 
After the input PSLG line segments are split, I will
refer to each of the individual split segments
as subsegments. 

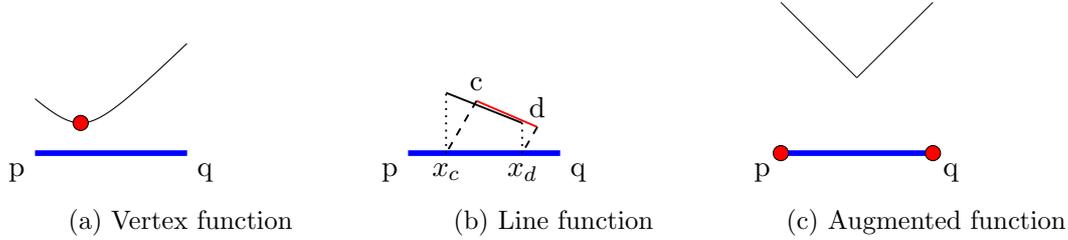
\begin{figure}
\centering
\begin{subfigure}[b]{0.3\textwidth}
\begin{tikzpicture}[scale=2]
    \draw [smooth,samples=100,domain=0:1] plot({\x},{sqrt((0.3-(\x))*(0.3-(\x)) + 0.2*0.2)});
    \draw [line width=0.75mm, blue] (0,0) -- (1,0);
    \draw [fill=red] (0.3,0.2) circle [radius=0.05cm];
    \node [below left] at (0,0) {p};
    \node [below right] at (1,0) {q};
\end{tikzpicture}
\caption{Vertex function}
\end{subfigure}
\centering
\begin{subfigure}[b]{0.3\textwidth}
\begin{tikzpicture}[scale=2]
    \draw [line width=0.75mm, blue] (0,0) -- (1,0);
    \draw [line width=0.25mm] (0.25,0.4) -- (0.75,0.2);
    \draw [line width=0.25mm, red] (0.45,0.3464) -- (0.85,0.1732);
    \draw [line width=0.25mm, dashed] (0.45,0.3464) -- (0.25,0.0);
    \draw [line width=0.25mm, dashed] (0.85,0.1732) -- (0.75,0.0);
    \draw [line width=0.25mm, dotted] (0.25,0.0) -- (0.25,0.4);
    \draw [line width=0.25mm, dotted] (0.75,0.0) -- (0.75,0.2);
    \node [below left] at (0,0) {p};
    \node [below right] at (1,0) {q};
    \node [above] at (0.45,0.3464) {c};
    \node [above] at (0.85,0.1732) {d};
    \node [below] at (0.25,0) {$x_c$};
    \node [below] at (0.75,0) {$x_d$};
\end{tikzpicture}
\caption{Line function}
\end{subfigure}
\centering
\begin{subfigure}[b]{0.3\textwidth}
\begin{tikzpicture}[scale=2]
    \draw [line width=0.75mm, blue] (0,0) -- (1,0);
    \draw [line width=0.1mm,  black] (0,1) -- (0.5,0.5); 
    \draw [line width=0.1mm,  black] (0.5,0.5) -- (1,1);
    \draw [fill=red] (0.0,0.0) circle [radius=0.05cm];
    \draw [fill=red] (1.0,0.0) circle [radius=0.05cm];
    \node [below left] at (0,0) {p};
    \node [below right] at (1,0) {q};
\end{tikzpicture}
\caption{Augmented function}
\end{subfigure}
\caption{The various distance functions associated with 
a line segment $pq$ in the PSLG. The blue, thick line segment $pq$
is horizontal and can be considered as part of the $x$ axis with
vertex $p$ being at the origin.  The distance to the red feature(s) 
is function of $x$, and the function is plotted as a thin black curve. 
(a) The distance to a PSLG vertex is plotted as a function of $x$.
The domain of the function is from $p$ to $q$.
(b) The distance to some other PSLG line segment (red) is plotted. 
The distance varies linearly, and  
the domain of the linear distance function is limited.  
The dashed lines are perpendicular to the PSLG line segment (red), 
and they limit the domain of the distance function.  
Beyond the domain, the distance to $c$ or $d$ (whichever is closer) 
defines the distance function on $pq$.  Those parts of the distance 
functions look like the distance function in (a). 
(c) As $p$ and $q$ are not adjacent,
the distance to $p$ or $q$, whichever is larger, also limits the 
feature size at any point on the line segment $pq$. A disk centered 
at a point on $pq$ with the radius equal to the greater of 
$xp$ or $xq$ contains both $p$ and $q$. Thus, this piecewise
linear function is also considered to compute the local
feature size.}
\label{fig:distance}
\end{figure}

\subsection{The Feature Size Function Computation}
The algorithm requires the knowledge of the local
feature size at every point on the input line segments
of the PSLG. In order to compute the feature
size, let the  $i^{\mathrm{th}}$ line 
segment of PSLG, $L_i$, be 
parameterized to lie on the $x$-axis from $x=0$ to
$x=l_i$, where $l_i$ is its length.  I compute the
piecewise smooth function $F(x)$ that provides
the feature size at any point $x$ on the line segment.  
I call this the feature size function of $L_i$.
As mentioned in the previous section, the local 
feature size at any point on $L_i$ 
is the distance to the nearest feature that
is not adjacent to $L_i$ or the distance to the
farthest end point of $L_i$, whichever is 
smaller.  

Clearly, $F(x)$ is the lower envelope
of many different functions, which plot
the distance to a vertex or a line (nonadjacent features)
in the PSLG from $x$.  Examples of such distance functions 
are shown in Fig.~\ref{fig:distance}.  The distance from a
point $(x,0)$ on $L_i$ to a vertex $(a,b)$ 
in the PSLG is 
$\sqrt{(a-x)^2 + b^2}$.  The function is shown in
Fig.~\ref{fig:distance}(a).  The distance 
between line segments $L_i$ (also denoted as $pq$)
to some other line segment $L_{cd}$, whose end points
are $c$ and $d$, is given by a linear function 
whose domain is
from $x_c$ to $x_d$, where $x_c$ and $x_d$ are
points on $L_i$ such that $x_cc$ and $x_dd$ are
perpendicular to $L_{cd}$ as shown in 
Fig.~\ref{fig:distance}(b).  From $x=0$ to $x=x_c$
and from $x=x_d$ to $x=l_i$, the distance to
vertex $c$ and $d$, respectively, defines the 
distance function.  
We should also consider the function in
Fig.~\ref{fig:distance}(c) that accounts for the
farthest end points of $L_i$ to the list of functions
over which we compute the lower envelope.  The
value of this function is $l_i$ at $x=0$ and $x=l_i$,
and it is $l_i/2$ at $x=l_i/2$.  
To compute the feature size function of $L_i$, 
we consider the distance function from
all vertices and line segments of the PSLG (including
end points of the line segments) and the farthest end
points of $L_i$.  
Note that the distance functions and
the feature size function need to be computed only
from $x=0$ to $x=l_i$.  

The lower envelope of the distance function 
can be computed using a sweep line
algorithm that maintains a balanced binary search 
tree or a heap that orders the functions based on their value
at the current location of the sweep line.  As I focus
on mesh generation in this paper, I direct
the readers to a paper~\cite{BGR97}
that solves the problem of computing the lower
envelope efficiently. The paper is
written to compute the lower envelope of lines and 
line segments, but it can be easily adapted 
for nonlinear distance functions in our context.

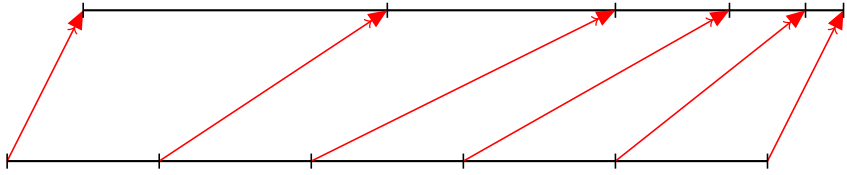
\begin{figure}
  \centering
  \begin{tikzpicture}
        \draw[thick] (0,0) -- (10,0);
        \draw[thick] (1,2) -- (11,2);
        \draw[-> {triangle 45},semithick,red] (0,0) -- (1,2);
        \draw[-> {triangle 45},semithick,red] (2,0) -- (5,2);
        \draw[-> {triangle 45},semithick,red] (4,0) -- (8,2);
        \draw[-> {triangle 45},semithick,red] (6,0) -- (9.5,2);
        \draw[-> {triangle 45},semithick,red] (8,0) -- (10.5,2);
        \draw[-> {triangle 45},semithick,red] (10,0) -- (11,2);
        \draw[semithick] (0,-0.1) -- (0,0.1);
        \draw[semithick] (2,-0.1) -- (2,0.1);
        \draw[semithick] (4,-0.1) -- (4,0.1);
        \draw[semithick] (6,-0.1) -- (6,0.1);
        \draw[semithick] (8,-0.1) -- (8,0.1);
        \draw[semithick] (10,-0.1) -- (10,0.1);
        \draw[semithick] (1,1.9) -- (1,2.1);
        \draw[semithick] (5,1.9) -- (5,2.1);
        \draw[semithick] (8,1.9) -- (8,2.1);
        \draw[semithick] (9.5,1.9) -- (9.5,2.1);
        \draw[semithick] (10.5,1.9) -- (10.5,2.1);
        \draw[semithick] (11,1.9) -- (11,2.1);
  \end{tikzpicture}
  \caption{An example of reference-to-PSLG mappings $M_i(t)$ from a
    reference segment $T_i$ to a PSLG segment $L_i$. 
    The reference segment is uniformly split into $n$ subsegments, and the
    corresponding splits are made in the PSLG segment.  The mapping function
    is defined such that uniform splits in the reference segment
    correspond to asymptotically proportional (to the local feature size) 
    splits in the PSLG segment.  Note that the reference segment and
    the PSLG segment may be of different lengths.}
  \label{fig:map}
\end{figure}
\subsection{The PSLG Segment Splitting}
After the feature size function is computed for each
line segment in the PSLG, the next step is to split
the line segments such that the length of each subsegment
is asymptotically proportional (see Section 3.3 for the
definition of asymptotic proportionality)
to the feature size at the end points of each 
subsegment.  In order to
achieve this goal, I construct a mapping function from
a reference segment to the PSLG line segment 
(see Fig.~\ref{fig:map}) such that
for each point on the reference segment, there is a
corresponding point on the PSLG segment (and vice 
versa).  When I split the reference segment evenly
and correspondingly split the PSLG line segment at
the mapped location, the mapping
function ensures that the length of the subsegments
in the PSLG line segment is asymptotically proportional 
to the local feature size.  In this section, I will explain how
the mapping function is computed by constructing and
solving a differential equation. I will also explain
how all the reference segments (there is one reference
segment for every line segment in the PSLG) are split 
such that the corresponding splits in the PSLG line 
segments are size optimal.

\subsubsection{Deriving the Differential Equations}
Let the mapping function for the $i^\mathrm{th}$ 
line segment $L_i$ be denoted by $M_i(t)$, where
$t$ is a point on the reference segment $T_i$. Let the
length of the reference segment be $(t^*)_i$, which is
yet to be determined.  The mapping function
should be designed such that $M_i(0)=0$ and 
$M_i((t^*)_i)=l_i$, where $l_i$ is the length of 
$L_i$.  Let $T_i$ be split into $n$ equal subsegments, 
which means that we split $L_i$ at $x=M_i(0)=0$, 
$x=M_i((t^*)_i/n)$, $x=M_i(2(t^*)_i/n)$, $x=M_i(3(t^*)_i/n)$, and
so on until $x=M_i((t^*)_i)=l_i$.  The length of
each split in the reference segment is given by
$h = (t^*)_i/n$.  

We want the mapping function to result in splits 
that are asymptotically proportional to the local
feature size.  Consider a vertex at $t$ on $T_i$.
Its corresponding point on $L_i$ is $M_i(t)$. The
vertex adjacent to $t$ on $T_i$ is $t+h$.  Its
corresponding point on $L_i$ is $M_i(t+h)$.  
The length of this subsegment on $L_i$ is 
$M_i(t+h)-M_i(t)$.  This length should be 
proportional to the feature size at $M_i(t)$, i.e.,
$F(M_i(t))$, where $F(\cdot)$ is the local feature
size function computed above by constructing the
lower envelope of the distance functions.  Thus,
$M_i(t+h)-M_i(t) \propto F(M_i(t))$. If $h$ is small
enough, we know that 
$(M_i(t+h)-M_i(t)) \approx hM_i^{'}(t)$.  Thus, my
intuition is to
compute $M_i$ such that $M_i^{'}(t)=F(M_i(t))$.  
In Section 5, I will show that this intuitive choice
of $M_i$ results in asymptotically proportional 
splits of the PSLG.
Note that the feature size function is always 
positive. Therefore, the mapping function is
monotonically increasing as its derivative is
also always positive. 

\subsubsection{Solving the Differential Equations}
Let us consider a line segment $L_i$ of the PSLG
and the corresponding reference segment $T_i$.  
Let the feature size function $F(\cdot)$ (obtained in
Section 4.1) on that 
line segment have $k$ parts, i.e., there are $k$ pieces
in the piecewise-smooth function.  As seen in Section 4.1, 
$F(\cdot)$ has parts that are either linear or 
the square root of a quadratic function.  
Let us denote $M_i$ as $y$. Our differential
equation is $y^{'} = F(y)$.  When a part of the
feature size function is linear, the 
equation becomes $y^{'} + ay = b$ for some $a$ and
$b$.  The solution (see Section 3.3) to this equation is 
$y = (b/a) + c/e^{at}$, where $c$ needs to be 
determined using a boundary condition.  
When the part of the feature size function 
is the square root of 
a quadratic function, the equation is of the form
$y^{'} = \sqrt{y^2 + 2ay + b}$ for some $a$ and $b$.  
Squaring both sides, we get
$(y^{'})^{2} = y^2 + 2ay + b$.  Differentiating w.r.t.
$t$, we get $2y^{'}y^{''} = 2yy^{'} + 2ay^{'}$. Dividing
by $2y^{'}$ on both sides, we get $y^{''} = y + a$.  
The solution(see Section 3.3) to this equation is 
$y = c_1e^t + c_2e^{-t} - a$, where $c_1$ and $c_2$
need to be determined using two boundary conditions.
As $F(\cdot)$ is piecewise smooth,
$y(t)$ is also piecewise smooth, and each part of 
$y$ is given by the solution above.  

In the solution to the differential equation provided 
in Section 3.3, there are constants that need to be 
evaluated based on the boundary conditions. 
Let us consider the first part of the feature size 
function along $L_i$.  
When the first part of $F(y)$ is linear, we 
use the value of the mapping function at the initial 
point. In our case, the initial value at $t=0$ is 
$y(0) = M_i(0) = 0$, i.e.,
the starting point on $T_i$ maps to the starting
point $L_i$. On the other 
hand, when $F(y)$ is of the form $\sqrt{y^2 + 2ay + b}$, 
the corresponding differential equation is of the form
$y^{''} =y + a$, which needs two initial value 
conditions.  The first one is $y(0)=0$ as above. The
second condition is given by the local feature size at $t=0$.
Since $y^{'}(t)$ is equal to the local feature size at
$y(t)$, our second boundary condition is $y^{'}(0) = F(0)$.  

Using the boundary value conditions provided above,
we can analytically compute the first part of the 
solution of the differential equation.  
Thus, we have computed the
first part of the mapping function from the reference
line segment to the PSLG line segment.  Let us assume that
the first part of the feature size function (on the actual
PSLG segment $L_i$) starts 
at $x=x_0=0$ and ends at $x=x_1$.  
The first part of the solution to the differential equation 
(on the reference segment $T_i$)
starts at $t=t_0=0$ and ends at $t=t_1$, where
$y(t_1) = x_1$.  Unfortunately, $t_1$ cannot be
analytically computed in all cases, and hence, it
needs to be numerically computed in a practical 
implementation.

Let us assume that the $j^{\mathrm{th}}$ part of the
feature size function starts at $x_{j-1}$ and ends
at $x_{j}$.  
For the second (and subsequent) parts of the solution
of the differential equation, the boundary
values are given by $y(t_j) = x_j$ and 
$y^{'}(t_j) = F(x_j)$ for $j>0$.  If there are $k$
parts of the feature size function, the length of
the reference segment $T_i$ is $(t^*)_i = t_k$,
where $y(t_k) = l_i$. 
Note that when the feature size is small, the length
of the reference segment is large because the
mapping function, whose derivative is proportional 
to the local feature size, grows slowly when the 
feature size is small. An example of the 
construction of the reference segment is provided
in the proof of Lemma~\ref{lemma:minlength}. 

\subsubsection{Splitting the Line Segments}
Now that I have computed the mapping functions from
every reference segment $T_i$ to the PSLG segment $L_i$,
our task is to split the reference segments evenly 
so that the PSLG segments are split asymptotically
proportional to local feature size.  I first split
the reference segment with the smallest length
into $n^*$ parts, where $n^*$ is determined by the 
satisfaction of the lemmas in Section 5
(specific lemmas are mentioned in Section 4.4).  
Let the length of the shortest reference segment
be $t^*_{\mathrm{min}}$. I then split the
reference segment $T_i$ into $n_i$ parts,
where  $$n_i = \floor*{n^* \frac{(t^*)_i}{t^*_{\mathrm{min}}}}$$ 
and $(t^*)_i$ is the length of $T_i$.
Note that when the feature size is small, 
the length of the reference segment is large,
and therefore, the number of splits is also large.  

\subsection{The Off-Center Vertex Insertion}
After the PSLG line segments are split into subsegments, 
I use prior algorithms to obtain a mesh with the desired
quality.  I use \"{U}ng\"{o}r et 
al.'s~\cite{EU09,U04} algorithm to place 
off-center vertices to eliminate skinny triangles
from the mesh.  As in their algorithm, I prioritize
skinny triangles with shortest edges.  To reiterate,
I will consider the shortest edge in every skinny
triangle and pick the triangle with the shortest 
edge among those considered edges. 
In the presence of small angles in the input PSLG, 
I use the algorithm by Pav et al.~\cite{MPW05,P03}
to decide which skinny triangles to ignore
because no algorithm can eliminate all of them.  
I will show that the splits in Section 4.2 
ensure that there is no vertex encroachment 
if the lemmas in the next section are satisfied.  

In order to obtain a high-quality mesh, a truly 
Delaunay triangulation or a constrained Delaunay 
triangulation (whichever is desired)
of the domain is constructed.  
Then a skinny triangle (if any) with the shortest 
edge is chosen. 
If the end points of the shortest edge of the skinny
triangle belong to line segments of the PSLG that
form a small angle (defined in Section 3.6), and if
the shortest edge is shorter than a certain threshold
(explained in Section 5.3), 
the skinny triangle is ignored.  
This skinny triangle is considered to be ``across''
a small angle. 
If not, its off-center point (see Section 3.4)
and the circumcenter of the skinny triangle are 
considered for insertion into the mesh.  Whichever
point is closer to the shortest edge of the skinny
triangle is inserted, and the domain is 
retriangulated. Another skinny triangle (if any)
with the shortest edge is chosen to be eliminated.  
These steps are repeated until all skinny triangles
are eliminated (except the ones 
across a small angle). 

\subsection{Satisfaction of Lemmas}
For truly Delaunay meshes, in the absence of
small angles, the PSLG segments should be split
such that the conditions in Lemmas~\ref{lemma:beginastar}
and~\ref{lemma:diacondition} are satisfied.  
In the presence of small angles,
in addition to satisfying the conditions in these lemmas,
the PSLG should be refined until the Delaunay
triangulation of the vertices on the PSLG segments
recover the PSLG segments. For constrained Delaunay
meshes, the conditions in Lemmas~\ref{lemma:beginastar},
~\ref{lemma:pslgskinnycondition}, and~\ref{lemma:diaskinny}
should be satisfied. For both truly and constrained
Delaunay meshes, in the presence of small angles,
as we progressively refine the PSLG, the bounds on 
the minimum and the maximum angle improve.

%% file: analysis0.tex
\section{An Analysis of the Algorithm}
I will recap some of the notations from Section 3
because I use them extensively in the analysis of
the algorithm.  In my analysis, I will first show that
the differential equation-based splitting of the input
PSLG line segments will result in size-optimal subsegments
such that $A^*\le\frac{\mathrm{LFS}(p)}{l}\le B^*$, where
$l$ is the length of a subsegment one of whose end points
is $p$, $\mathrm{LFS}(\cdot)$ and $F(\cdot)$ are used
to denote the local feature size
function, and $A^*$ and $B^*$ are some constants. 
I will denote the ratio $B^*:A^* \ge 1$ as $R$. A skinny 
triangle is any triangle whose minimum angle is
$\theta<\theta^*<\frac{\pi}{6}$, where $\theta^*$ 
is provided as an
input to the algorithm. I will then derive conditions
such that no vertex encroaches upon a PSLG subsegment.
The conditions are a function of $A^*$, $B^*$, $R$, 
$\theta^*$, and $\alpha = 1/(2\sin{(\theta^*)})$, where 
$\alpha$ is the desired minimum radius-edge ratio of 
triangles in the mesh. I will also show that the algorithm
terminates with a size-optimal mesh. In the bound 
associated with the size optimality, 
$\beta = 1/(2\sin{(\theta^*/2)})$ is the maximum distance
(normalized to the length of the shortest edge of a skinny
triangle) at which a Steiner vertex is placed from the
vertices of the shortest edge of a skinny triangle. 
A small angle
$\phi$ is any angle in the input that is 
$\phi\le\arccos{\frac{1}{2R}}$.  In the initial analysis, I
will consider any angle $\phi\le\pi/2$ as a small angle.  In the
appendix, I will show why $\phi$ can be smaller.  
Finally, I will show that
even in the presence of skinny triangles, as long as
we refine the PSLG sufficiently, it is possible to obtain
truly or constrained Delaunay meshes such that
the maximum angle is bounded. 

\subsection{Splitting the PSLG Segments}

In the first few lemmas below, I will show that 
as the PSLG segments are progressively refined, 
the bounds $A^*$ and $B^*$ increase, but their
ratio $R$ approaches $1$.  In addition, I will show
that given an upper bound on $R$ or a lower
bound on $A^*$, it is possible to split the PSLG
such that $B^*$ is bounded from above. 

\begin{lemma}
\label{lemma:minlength}
The length of the shortest reference segment
$t^*_{\mathrm{min}}\ge2\log_e{2}$.
\end{lemma}

\begin{proof}
In order to obtain the shortest reference segment,
the local feature size at any point on the PSLG
line segment $pq$ of length $l$ has to be as large
as possible. The feature size cannot be arbitrarily 
large even if other features in the PSLG are very 
far away because the feature size is
bounded from above by $l-x$ from $x=0$ to $x=l/2$
and by $x$ from $x=l/2$ to $x=l$ (see 
Fig.~\ref{fig:distance}(c)), if $pq$ is 
assumed to be on the x-axis and $p$ is at the origin
(no loss of generality).

Let us now derive the mapping function to compute
the length of the reference segment.  The 
differential equation for the first piece of
the mapping function is given by
$y^{'} = l - y$, where $y(t)$ is the mapping function,
which implies $y+y^{'}=l$.  The solution to 
this equation is given by $y = l+c/e^t$, where
$c$ is a constant.  When $t=0$, $y=0$. Therefore,
$0 = l + c$, which implies $c=-l$. Thus, the first
piece of the mapping function is $y(t) = l - l/e^t$.  
This piece spans from $t=0$ to $t=t_1$ such that
$y(t_1) = l/2$ because the first piece of
the differential equation spans from $x=0$ to $x=l/2$,
which implies $l - l/e^{t_1} = l/2$,
which implies $t_1=\log_e{2}$.  As the other half of 
the local feature size function is symmetric, the 
length of the reference line segment is at least 
$2\log_e{2}$.  Thus, $t^*_{\mathrm{min}}\ge2\log_e{2}$.  
\end{proof}

The following lemma applies only to the PSLG segment
with the shortest reference segment because the variables
$t^*_{\mathrm{min}}$ and $n^*$ pertain to the PSLG
segment. If the two variables are replaced with their
equivalent quantities for other segments, the lemma also
holds for other segments. 

\begin{lemma}
\label{lemma:minsplit}
If the PSLG segment $L_i$ with the shortest reference segment
whose length is
$t^*_{\mathrm{min}}$
is split into $n^*$ subsegments, the bound on the ratio
of the local feature size and length of a subsegment
on $L_i$
at some vertex $p$ is given by 
$A^*\le\frac{\mathrm{LFS}(p)}{l_p}\le B^*$, where
$l_p$ is the length of a subsegment one of whose end points
is $p$, $\mathrm{LFS}(\cdot)$ is the local feature size
function, and 
$$A^* = \frac{n^*}{t^*_{\mathrm{min}}} - 1
\mathrm{\ and\ } 
B^* = \frac{n^*}{t^*_{\mathrm{min}}} + 1.$$ 
\end{lemma}

\begin{proof}
Let us denote the PSLG line segment under consideration
as $L_i$. In the algorithm, I 
split the corresponding reference segment $T_i$ 
into $n^*$ equal parts.  The length of each part of the
reference segment is 
$h = t^*_{\mathrm{min}}/n^*$. Let $M_i(t)$ be the mapping
function.  Let us assume that $t = t_p$ is a vertex on $T_i$ 
such that it maps to vertex at $x=p$ on $L_i$, i.e., 
$M_i(t_p) = p$.  
The length of one of the segments 
\footnote{Note that the lemma 
also holds for the other segment, but I omit that case since 
the proofs are identical.}
at $p$ is given by
$l_p = M_i(t_p+h) - M_i(t_p)$ because a vertex next to
$t_p$ on $T_i$ is at $t_p+h$.
By the mean value theorem,
$$l_p = M_i(t_p + h) - M_i(t_p) = hM^{'}_i(t_p + h_0),$$
where $0\le h_0\le h$ is some constant.  Since 
$M^{'}_i(t)$ is the local feature size function 
$F(M(t))$ (also denoted as $\mathrm{LFS}(M(t))$),
\begin{equation}\label{eq:LFS_mvt}
l_p = M_i(t_p + h) - M_i(t_p) = hF(M_i(t_p + h_0)).
\end{equation}
Note that $M_i(t_p)$ corresponds to vertex $p$ on the
PSLG segment, and let $M_i(t_p+h)$ correspond to
vertex $q$ on the PSLG segment.  With these vertices,
we can apply the property of Lipschitz functions
in the next step. 
Since the local feature size function $F(\cdot)$ is 
a Lipschitz function (see Section 3.2), 
$$F(M_i(t_p + h_0)) \le F(M_i(t_p)) + |(M_i(t_p+h_0) - M_i(t_p))|$$
and 
$$F(M_i(t_p + h_0)) \ge F(M_i(t_p)) - |(M_i(t_p+h_0) - M_i(t_p))|.$$
Since the mapping function  
is a monotonically increasing function,
$$F(M_i(t_p + h_0)) \le F(M_i(t_p)) + (M_i(t_p+h_0) - M_i(t_p))$$
and 
$$F(M_i(t_p + h_0)) \ge F(M_i(t_p)) - (M_i(t_p+h_0) - M_i(t_p)).$$
Substituting the inequalities above into Eq.~\ref{eq:LFS_mvt},
we get
$$l_p = hF(M_i(t_p + h_0)) \le h(F(M_i(t_p)) + (M_i(t_p+h_0) - M_i(t_p)))$$
and
$$l_p = hF(M_i(t_p + h_0)) \ge h(F(M_i(t_p)) - (M_i(t_p+h_0) - M_i(t_p))).$$
Also since $(M_i(t_p+h_0) - M_i(t_p)) \le l_p$ (because $0\le h_0\le h$),
we get
$$l_p \le h(F(M_i(t_p)) + l_p)$$
and 
$$l_p \ge h(F(M_i(t_p)) - l_p).$$
Substituting $M_i(t_p)$ with $p$ and rearranging,
$$l_p(1-h) \le h F(p) \le l_p (1+h),$$
which implies
$$\frac{(1-h)}{h} \le \frac{F(p)}{l_p} \le \frac{(1+h)}{h}.$$
Thus,
\begin{equation}\label{eq:tmina}
A^* = \frac{(1-h)}{h} = \frac{1 - t^*_{\mathrm{min}}/n^*}{t^*_{\mathrm{min}}/n^*} = \frac{n^* - t^*_{\mathrm{min}}}{t^*_{\mathrm{min}}} = \frac{n^*}{t^*_{\mathrm{min}}} - 1
\end{equation}
and
\begin{equation}\label{eq:tminb}
B^* = \frac{(1+h)}{h} = \frac{1 + t^*_{\mathrm{min}}/n^*}{t^*_{\mathrm{min}}/n^*} = \frac{n^* + t^*_{\mathrm{min}}}{t^*_{\mathrm{min}}} = \frac{n^*}{t^*_{\mathrm{min}}} + 1.
\end{equation}
\end{proof}

The following lemma applies to all PSLG
segments.  In the proof, I will substitute 
$t^*_{\mathrm{min}}$ and $n^*$ (seen in the
lemma above) with their equivalents,
$t^*_i$ and $n$, respectively, for any PSLG
segment. I have explicitly mentioned about 
the substitution here so that there is no 
confusion about the lemmas. 

\begin{lemma}
\label{lemma:allsplit}
If the $i^\mathrm{th}$ PSLG segment with a
reference segment of length $t^*_i$
is split into $n = \floor{n^* \frac{t^*_i}{t^*_{\mathrm{min}}}}$ 
subsegments, the bound on the ratio
of the local feature size and length of the subsegment
at some vertex $p$ is given by 
$A^*\le\frac{\mathrm{LFS}(p)}{l_p}\le B^*$, where
$l_p$ is the length of a subsegment one of whose end points
is $p$, $\mathrm{LFS}(\cdot)$ is the local feature size
function, and 
$$A^* = \frac{n^*}{t^*_{\mathrm{min}}} - \frac{1}{2\log_e{2}}-1 
\mathrm{\ and\ } 
B^* = \frac{n^*}{t^*_{\mathrm{min}}} + 1.$$
\end{lemma}

\begin{proof}
In the proof of Lemma~\ref{lemma:minsplit}, 
Eq.~\ref{eq:tmina} and Eq.~\ref{eq:tminb} were
obtained without any assumption about the length of the reference segment
being the shortest.  Therefore, the lower and upper bounds 
$A^*$ and $B^*$  for $\frac{\mathrm{LFS}(p)}{l_p}$ are given by
$$A^* = \frac{n}{t^*_i}-1 \mathrm{\ and\ } B = \frac{n}{t^*_i} + 1.$$
Substituting $n = \floor{n^* \frac{t^*_i}{t^*_{\mathrm{min}}}}$
in the above equation, 
$$A^* = \frac{\floor{n^* \frac{t^*_i}{t^*_{\mathrm{min}}}}}{t^*_i}-1 
\mathrm{\ and\ } 
B^* = \frac{\floor{n^* \frac{t^*_i}{t^*_{\mathrm{min}}}}}{t^*_i} + 1.$$
This equation can be rewritten as
$$A^* = \frac{n^* \frac{t^*_i}{t^*_{\mathrm{min}}} - \epsilon}{t^*_i}-1 
\mathrm{\ and\ } 
B^* = \frac{n^* \frac{t^*_i}{t^*_{\mathrm{min}}} - \epsilon}{t^*_i} + 1,$$
where $\epsilon<1$.  This equation is equivalent to
$$A^* = \frac{n^*}{t^*_{\mathrm{min}}} - \frac{\epsilon}{t^*_i}-1 
\mathrm{\ and\ } 
B^* = \frac{n^*}{t^*_{\mathrm{min}}} - \frac{\epsilon}{t^*_i}+1.$$
Since $0\le\epsilon<1$ and $t^*_i \ge 2\log_e{2}$ (by Lemma~\ref{lemma:minlength}),
$$A^* \ge \frac{n^*}{t^*_{\mathrm{min}}} - \frac{1}{2\log_e{2}}-1 
\mathrm{\ and\ } 
B^* \le \frac{n^*}{t^*_{\mathrm{min}}} + 1.$$
Thus,
$$A^*\le\frac{\mathrm{LFS}(p)}{l_p}\le B^*,$$ where
$$A^* = \frac{n^*}{t^*_{\mathrm{min}}} - \frac{1}{2\log_e{2}}-1
\mathrm{\ and\ } 
B^* = \frac{n^*}{t^*_{\mathrm{min}}} + 1.$$
\end{proof}

In the next two lemmas, I will show that it is
possible to split the PSLG segments such that
there is an upper bound $B^*$ on 
$\frac{\mathrm{LFS}(p)}{l_p}$
given a lower bound $A^*$ or an upper
bound on the ratio $R$, where $p$ is
a vertex on the PSLG line segment, and $l_p$ is
the length of a subsegment at $p$. 

\begin{lemma}
\label{lemma:aboundsb}
Given a lower bound $A^*$ on $\frac{\mathrm{LFS}(p)}{l_p}$,
it is possible to split the PSLG line segments such that the
upper bound $B^* \le A^* + 1/\log_e{2} + 2$, where $p$ is
a vertex on the PSLG line segment, and $l_p$ is
the length of a subsegment at $p$. 
\end{lemma}

\begin{proof}
In Lemma~\ref{lemma:allsplit}, we showed that if the shortest
reference segment is split into $n$ equal parts, the lower bounds 
on $\frac{\mathrm{LFS}(p)}{l_p}$ on any PSLG segment are
given by
$$A^* = \frac{n}{t^*_{\mathrm{min}}} - \frac{1}{2\log_e{2}}-1,$$ 
which is linear in $n$.
From this equation, given a lower bound $A^*$, it is possible to 
compute the minimum $n$ to obtain the lower bound by 
solving the linear equation.  But this $n$ might not
be an integer.  Therefore, we choose the ceiling of $n$, 
$\ceil{n}$, i.e., we split the shortest reference segment
into $\ceil{n}$ subsegments in order to obtain the lower 
bound $A^*$ on the ratio of $\mathrm{LFS(p)}$ and $l_p$. 
Thus, our bound $A^*$ increases to
$$A^* = \frac{n+\epsilon}{t^*_{\mathrm{min}}} - \frac{1}{2\log_e{2}}-1,$$ 
where $0\le\epsilon<1$. Therefore,
\begin{equation}\label{eq:alow}
A^* \ge \frac{n}{t^*_{\mathrm{min}}} - \frac{1}{2\log_e{2}}-1. 
\end{equation}
Similarly, for the upper bound (from Lemma~\ref{lemma:allsplit}),
$$B^* = \frac{n+\epsilon}{t^*_{\mathrm{min}}}+1 
= \frac{n}{t^*_{\mathrm{min}}} +\frac{\epsilon}{t^*_{\mathrm{min}}}+1.$$ 
Since $t^*_{\mathrm{min}} \ge 2\log_e{2}$ (Lemma~\ref{lemma:minlength})
and $0\le\epsilon<1$, 
\begin{equation}\label{eq:bhigh}
B^* < \frac{n}{t^*_{\mathrm{min}}} + \frac{1}{2\log_e{2}} + 1. 
\end{equation}
If we split the shortest reference segment into $\ceil{n}$
subsegments and $A^*$ is as small as it can be 
(Eq.~\ref{eq:alow}) and $B^*$ is as large as it can 
be (Eq.~\ref{eq:bhigh}), the difference between them 
is, at most, $1/\log_e{2} + 2$.  
Therefore, given a lower bound $A^*$, we can split the PSLG 
line segments such that the upper bound $B^*$ is at most 
$A^* + 1/\log_e{2} + 2$.
\end{proof}

\begin{lemma}
\label{lemma:rboundsb}
Given an upper bound on the ratio $R>1$ of the upper and lower
bound on $\frac{\mathrm{LFS}(p)}{l_p}$, where $p$ 
is a vertex on the PSLG line segment and $l_p$ is
the length of a subsegment at $p$,
it is possible to split the PSLG line segments such that
there is an upper bound $B^*$ that is only a function
of $R$.
\end{lemma}

\begin{proof}
In Lemma~\ref{lemma:allsplit}, we showed that if the shortest
reference segment is split into $n$ equal parts, the lower bounds 
on $\frac{\mathrm{LFS}(p)}{l_p}$ on any PSLG segment are 
given by
$$A^* = \frac{n}{t^*_{\mathrm{min}}} - \frac{1}{2\log_e{2}}-1,$$ 
and the upper bound is given by
$$B^* = \frac{n}{t^*_{\mathrm{min}}} + 1 = A^* + \frac{1}{2\log_e{2}} + 2.$$ 
The ratio $R = B^*/A^*$ is given by
$$R = 1 + \frac{1}{2A^*\log_e{2}} + \frac{2}{A^*}.$$
Clearly, it is possible to compute $A^*$ (as a function of only 
$R$) and the corresponding $n$ for which $B^*/A^* = R$. 
The computed $A^*$ may not correspond 
to $n$ being an integer, so we take the ceiling $\ceil{n}$. 
This operation can only decrease $R$ because as $A^*$ tends
to infinity, $R$ tends to 1.  As we saw in 
Lemma~\ref{lemma:aboundsb}, the upper bound $B^*$ is bounded 
for a given lower bound $A^*$.  As we have computed 
the required lower bound $A^*$ for a given $R$, it is
possible to split the PSLG segments such that the 
upper bound $B^*$ is also bounded as a function of $R$.
\end{proof}

%% file: analysis1.tex
\subsection{Conditions for No Encroachment}
In Lemmas~\ref{lemma:aboundsb} and~\ref{lemma:rboundsb},
I have shown that the PSLG line
segments can be split in a size-optimal manner if a minimum
$A^*$ or a maximum $R>1$ is given. In the next set of lemmas,
I will derive the condition on $A^*$ and $R$ (as a function
of the minimum desired angle $\theta^*$ or the radius-edge ratio
$\alpha = 1/2\sin{\theta^*}$) such that
there is no encroachment of Steiner vertices upon
PSLG subsegments.  I have already shown that
for any such condition, it is possible to obtain a
size-optimal split. Further, I show that the final mesh
is also size optimal. For now, I assume that there are
no small angles in the PSLG.  In the next subsection,
I will analyze what happens when small angles are present
in the PSLG.

\subsubsection{Truly Delaunay Refinement}
I will first consider mesh refinement that gives us truly 
Delaunay meshes.  In order to obtain
truly Delaunay meshes, one should construct the Delaunay 
triangulation of the vertices inserted by our segment 
splitting algorithm and recover the PSLG, i.e., the Delaunay
triangulation should automatically include all the segments
in the PSLG.  To achieve this, note that if the 
diametral circle of every subsegment is empty,
the PSLG is recovered by the Delaunay triangulation of
the vertices. To see why, consider the midpoint
of any subsegment. Its nearest vertices are the two vertices
forming the subsegment.  Thus, the two vertices are neighbors 
in the Voronoi diagram, which is the dual of the Delaunay 
triangulation.

I will first provide the condition on $A^*$ and
$R$ for which the PSLG is recovered.  I will then provide
the conditions on $A^*$ for which no Steiner vertices
will be inserted in the diametral circle of the PSLG
subsegments, which ensures that our algorithm terminates
with a size-optimal, high-quality mesh.

\begin{lemma}
\label{lemma:beginastar}
If $A^*>1/\sqrt{2}$, the diametral circle of any
PSLG subsegment (after the split) does not
contain vertices from nonadjacent segments.
\end{lemma}

\begin{proof}
If a subsegment $ab$ is of length $l$ and $A^*>1/\sqrt{2}$,
the feature size at $a$ and $b$ is at least $l/\sqrt{2}$.
Thus, the nearest nonadjacent feature is at least
$l/\sqrt{2}$ away. The farthest point from $a$ or $b$
inside the diametral circle of $ab$ is at a distance
$l/\sqrt{2}$ away. Thus, the diametral circle does
not contain any vertices from nonadjacent features.
\end{proof}

For now, I will assume that a small angle in the PSLG is
any angle less than $\pi/2$. In the appendix, I will show
why it is possible to lower the threshold for a small 
angle to $\arccos{(1/2R)}$.
This small change has limited implications
in the proofs in the rest of the paper.

\begin{lemma}
If $A^*\ge 1/\sqrt{2}$ and the minimum angle 
$\phi \ge \pi/2$, 
the PSLG is recovered by the Delaunay triangulation of
the vertices on subsegments.
\end{lemma}

\begin{proof}
Since vertices on nonadjacent segments are not inside
the diametral circle of the PSLG subsegments, and 
since $\phi \ge \pi/2$, the diametral circles of
the PSLG subsegment are empty.  Thus,
the PSLG is recovered by the Delaunay triangulation.
\end{proof}

Before I proceed, I will define what I mean by 
the $k^{\mathrm{th}}$ layer (layer of order $k$)
of the advancing 
front of Steiner vertices. Let the shortest subsegment
of the split PSLG be of length $l_0$. The vertices 
on the shortest subsegment of the split PSLG are 
considered to be a part of 
the $0^{\mathrm{th}}$ layer of
the vertices.  For other vertices on
the PSLG, consider the shortest PSLG subsegment
adjacent to the vertex.  Let the length of the
subsegment be $l$. If $\alpha^{k-1} l_0 < l \le  
\alpha^{k} l_0$, the vertex is considered to a part
of the $k^\mathrm{th}$ layer. 
Note that if there are small angles (see Fig.~\ref{fig:smallangleproblem})
in the PSLG, the shortest edge adjacent to a vertex 
on a PSLG subsegment may be a very short edge
connecting it to a vertex on an adjacent PSLG segment. 
Assigning an order to such vertices can be confusing,
which is why we assume no small angles are present in 
this subsection. In the next section, I will
elaborate on how to assign their order.

When we insert a Steiner vertex (in the interior, 
not on a PSLG subsegment) $a$ into the mesh, we place it at a
distance of $l_a$ from the nearest 
vertex in the mesh (at the time of insertion). If 
$\alpha^{k-1} l_0 < l_a \le \alpha^{k} l_0$,
I consider $a$ to be part of the $k^{\mathrm{th}}$
layer of vertices. 

I will define the parent of a Steiner vertex $a$ 
as one of the vertices of the shortest side
of the skinny triangle $t$. The vertex $a$ is
either the off-center vertex of $t$ or its circumcenter. 
Between the two possible
vertices of $t$, pick the vertex that is on the
lower-order layer. 
If we move from the Steiner vertex to its parent, 
then to its grandparent, 
and so on, we will reach a vertex on a subsegment of
the PSLG. Let us call this vertex the ancestral
vertex of the Steiner vertex. Si~\cite{S09}
analyzed Shewchuk's 3D algorithm~\cite{S97} through
a similar sequence of vertices.  
As we move from a Steiner
vertex to its ancestral vertex, I will show 
below that the order of vertices monotonically
reduces.  Note that we may skip layers as
we move from a vertex to one of its children. 

\begin{lemma}
\label{lemma:parentlayer}
The order of a vertex is greater than that of
its parent. 
\end{lemma}

\begin{proof}
If the length of the shortest segment $s$
of a skinny triangle $t$ is $l$, one (or
both) of the vertices on the shortest segment
is at most on the layer $\ceil{\log_{\alpha}{(l/l_0)}}$. 
Why? Because the vertex of $s$ that was inserted
later was inserted at a distance of at most $l$
from other vertices in the mesh. 
The Steiner vertex $a$ that is inserted to 
replace $t$ is at least at a distance of 
$\alpha l$ from all other vertices in the mesh 
(because we prioritize skinny triangles with 
shortest edges first).  
Thus, $a$ belongs to a layer whose order is at 
least $\ceil{\log_{\alpha}{(\alpha l/l_0)}} = 1 +
\ceil{\log_{\alpha} (l/l_0)}$, which 
is greater than the order of its parent.
\end{proof}

\begin{lemma}
\label{lemma:layerk}
If a Steiner vertex $b$ belongs to layer
$k_b$ and its ancestral vertex $a$ on the PSLG
belongs to layer $k_a$, the upper bound on the
local feature size at $b$ is given by
$l_0\alpha^{k_a} (B^* + \alpha + \alpha^2 + ... + \alpha^{k_a-k_b}),$
where $l_0$ is the length of the shortest subsegment in
the split PSLG.  
\end{lemma}

\begin{proof}
Consider the path from $a$ to $b$ such that
every vertex is preceded by its parent. The total 
length of this path is maximized when the path
contains as many vertices as possible and when
each edge (from one vertex to the next) has the
maximum possible length. This maximization happens 
when Steiner
vertices do not skip a layer (which translates to
having as many vertices as possible in the path). 
Let the path be
$a$, $a_1$, $a_2$,...,$a_n$, $b$. Let the length
of the shortest PSLG subsegment at $a$ be $l_a$. 
Note that $A^* l_a \le f_a \le B^* l_a$. Since the layer
order increases by at least $1$, $a_1$ belongs
to layer $k_a + 1$, $a_2$ belongs to layer
$k_a + 2$, and so on. Also, $|aa_1| \le \alpha l_a$ 
because it belongs to layer $k_a+1$, $a_1a_2 \le
\alpha^2 l_a$, and so on. Thus, the maximum length
of the path is
$l_a(\alpha + \alpha^2 + ... \alpha^{k_b-k_a})$. 
The local feature size at $b$ is bounded by
$$f_b \le f_a + l_a(\alpha + \alpha^2 + ... \alpha^{k_b-k_a}),$$
where $f_a$ is the feature size at $a$ and $f_b$ is
the feature size at $b$. 
Since $f_a \le B^* l_a$,
$$f_b \le B^* l_a + l_a(\alpha + \alpha^2 + ... \alpha^{k_b-k_a}).$$
Since $l_a \le \alpha^{k_a} l_0$ (it belongs to layer $k_a$), 
where $l_0$ is the length
of the shortest side of the PSLG, 
$$f_b \le l_0\alpha^{k_a} (B^* + \alpha + \alpha^2 + ... \alpha^{k_b-k_a}).$$
\end{proof}

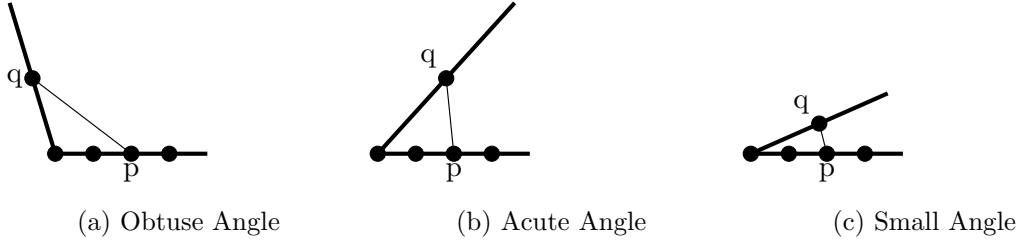
\begin{figure}
\centering
\begin{subfigure}[b]{0.3\textwidth}
\begin{tikzpicture}[scale=2]
    \draw [line width=0.55mm] (0,0) -- (1,0);
    \draw [line width=0.55mm] (0,0) -- (-0.3,1);
    \draw [line width=0.15mm] (0.5,0) -- (-0.15,0.5);
    \draw [fill=black] (0,0) circle (0.05);
    \draw [fill=black] (0.5,0) circle (0.05);
    \draw [fill=black] (0.25,0) circle (0.05);
    \draw [fill=black] (0.75,0) circle (0.05);
    \draw [fill=black] (-0.15,0.5) circle (0.05);
    \node [below] at (0.5,0) {p};
    \node [left] at (-0.15,0.5) {q};
\end{tikzpicture}
\caption{Obtuse Angle}
\end{subfigure}
\centering
\begin{subfigure}[b]{0.3\textwidth}
\begin{tikzpicture}[scale=2]
    \draw [line width=0.55mm] (0,0) -- (1,0);
    \draw [line width=0.55mm] (0,0) -- (0.9,1);
    \draw [line width=0.15mm] (0.5,0) -- (0.45,0.5);
    \draw [fill=black] (0,0) circle (0.05);
    \draw [fill=black] (0.5,0) circle (0.05);
    \draw [fill=black] (0.25,0) circle (0.05);
    \draw [fill=black] (0.75,0) circle (0.05);
    \draw [fill=black] (0.45,0.5) circle (0.05);
    \node [below] at (0.5,0) {p};
    \node [above left] at (0.45,0.5) {q};
\end{tikzpicture}
\caption{Acute Angle}
\end{subfigure}
\centering
\begin{subfigure}[b]{0.3\textwidth}
\begin{tikzpicture}[scale=2]
    \draw [line width=0.55mm] (0,0) -- (1,0);
    \draw [line width=0.55mm] (0,0) -- (0.9,0.4);
    \draw [line width=0.15mm] (0.5,0) -- (0.45,0.2);
    \draw [fill=black] (0,0) circle (0.05);
    \draw [fill=black] (0.5,0) circle (0.05);
    \draw [fill=black] (0.25,0) circle (0.05);
    \draw [fill=black] (0.75,0) circle (0.05);
    \draw [fill=black] (0.45,0.2) circle (0.05);
    \node [below] at (0.5,0) {p};
    \node [above left] at (0.45,0.2) {q};
\end{tikzpicture}
\caption{Small Angle}
\end{subfigure}
\centering
\caption{
Images depicting the problem with small angles
in the PSLG. The thick lines are part of the PSLG. 
The black dots are the vertices added to the 
PSLG (not all are shown).
(a) When the angle is obtuse, the line segment
$pq$ is longer than subsegments at $p$ and $q$. (b) When the
angle is acute, but greater than 60 degrees, $pq$ might be
longer than subsegments at $p$ and $q$, but it is not
guaranteed unless the segments are adequately refined  
(see the appendix for a detailed explanation). (c) When the angle
is very small, there is a good chance that $pq$ might be
shorter than the threshold for size optimality.  
}
\label{fig:smallangleproblem}
\end{figure}

Note that Lemmas~\ref{lemma:parentlayer} 
and~\ref{lemma:layerk} above do not hold when small
angles are present in the input because there
may be arbitrarily short edges that join a vertex $p$
from one segment to a vertex $q$ in its adjacent 
segment (see Fig.~\ref{fig:smallangleproblem}). 
If that edge, however, is longer than $\mathrm{LFS}(p)/B^*$
and $\mathrm{LFS}(q)/B^*$, the lemma still holds because
a child of $p$ or $q$ may simply skip a few layers if
the child is inserted due to a skinny triangle
with the shortest edge $pq$. 

In the next lemma, I derive the conditions 
that ensure that no Steiner 
vertices are added inside the diametral circles of
PSLG subsegments. These conditions ensure that
no PSLG subsegments are ever encroached upon by 
a Steiner vertex.

\begin{lemma}
\label{lemma:diacondition}
Let $\theta^*$ be the desired minimum angle in a mesh
and $\alpha = 1/(2\sin{\theta^*})$ be the desired maximum
radius-edge ratio.  If
$$\frac{B^*}{A^*} + \frac{\alpha}{A^*(\alpha-1)} + \frac{2}{A^*} \le \sqrt{2},$$
no Steiner vertices will be placed in  the diametral 
circles of any PSLG subsegments. 
\end{lemma}

\begin{proof}
If a Steiner vertex is to be placed inside the diametral
circle of a PSLG subsegment of length $l$, it should be at
a distance of less than $l/\sqrt{2}$ from one of the vertices
of the subsegment. This proof will derive the condition for which
it is impossible to place a Steiner vertex at such a short distance.

Let us assume that a Steiner vertex $a$ (of layer $k+1$, $k\ge 0$)
is placed ``close'' to a PSLG subsegment of length $l$.  
Since $a$ is placed at a distance of at least $\alpha^{k} l_0$
from all other vertices, let us assume that it is placed
at a distance $\gamma_1 \alpha^{k} l_0$ from the nearest
vertex, where $1<\gamma_1\le\alpha$ (the vertex $a$ is
on layer $k+1$). Let us also assume that it is at a distance
$\gamma_1\gamma_2\alpha^{k} l_0$ from vertex $p$ of the 
PSLG subsegment, where $\gamma_2 \ge 1$. 
The bound on the feature size at $a$ is given by
(by modifying Lemma~\ref{lemma:layerk} slightly and setting 
$k_a=0$ in the lemma)
$$f_a \le l_0 (B^* + \alpha + \alpha^2 + ... + \alpha^{k} + \gamma_1\alpha^{k}).$$
The bound on the LFS at vertex $p$ is given by
$$f_p \le l_0 (B^* + \alpha + \alpha^2 + ... + \alpha^{k} + \gamma_1\alpha^{k} + \gamma_1\gamma_2\alpha^{k}).$$
The length of the longest subsegment on the PSLG adjacent to $p$ is
bounded from above by $f_p/A^*$.  If this length is less than
$\sqrt{2}\gamma_1\gamma_2\alpha^{k} l_0$, then $a$
is not in the diametral circle of the PSLG subsegment. 
Thus, we want
$$\frac{l_0 (B^* + \alpha + \alpha^2 + ... + \alpha^{k} + \gamma_1\alpha^{k} + \gamma_1\gamma_2\alpha^{k})}{A^*} \le \sqrt{2}\gamma_1\gamma_2\alpha^{k} l_0.$$
After canceling $l_0$ on both sides, the inequality can be rewritten as
\begin{equation}
\label{eqn:sqrt2}
\frac{1}{\gamma_1\gamma_2}\frac{1}{\alpha^{k}} \frac{(B^* + \alpha + \alpha^2 + ... + \alpha^{k})}{A^*} + 
\frac{\gamma_1\alpha^k}{\gamma_1\gamma_2\alpha^k A^*} + \frac{\gamma_1\gamma_2\alpha^k}{\gamma_1\gamma_2\alpha^k A^*} \le \sqrt{2}.
\end{equation}
The inequality simplifies to
\begin{equation}
\frac{1}{\gamma_1\gamma_2}\frac{1}{\alpha^{k}} \frac{(B^* + \alpha + \alpha^2 + ... + \alpha^{k})}{A^*} + 
\frac{1}{\gamma_2 A^*} + \frac{1}{A^*} \le \sqrt{2}.
\end{equation}
The LHS of the above equation is maximized when $\gamma_1 = \gamma_2 =1$.
Therefore, if
$$\frac{1}{\alpha^{k}} \frac{(B^* + \alpha + \alpha^2 + ... + \alpha^{k})}{A^*} + \frac{2}{A^*} \le \sqrt{2},$$
$a$ is not inside the diametral circle of the PSLG subsegments at $p$.
This expression translates to
$$\frac{B^*}{\alpha^k A^*} + \frac{1}{A^*} 
\left(\frac{1}{\alpha^{k-1}} + \frac{1}{\alpha^{k-2}}  + ... + 1\right) 
+ \frac{2}{A^*} \le \sqrt{2}.$$
The first term of the LHS is maximized when $k=0$, and the second term
is maximized when the geometric progression extends to infinity. Thus,
if 
$$\frac{B^*}{A^*} + \frac{1}{A^*} 
\left(\frac{1}{1 - \frac{1}{\alpha}}\right) 
+ \frac{2}{A^*} \le \sqrt{2}.$$
$a$ will not be in the diametral circle of subsegments at $p$. 
This expression translates to
$$\frac{B^*}{A^*} + \frac{\alpha}{A^*(\alpha-1)} + \frac{2}{A^*} \le \sqrt{2},$$
which proves the lemma. 
\end{proof}

It is possible to refine our mesh such that the condition in 
Lemma~\ref{lemma:diacondition} above holds.  
As we refine the PSLG segments, $A^*$ tends to infinity and 
$B^*/A^*$ tends to 1. Thus, for any value on the RHS greater than $1$, 
it is possible to refine the mesh such that the inequality is satisfied. 
We will prove it more formally below. 

\begin{theorem}
\label{theorem:truesizeoptimal}
The algorithm terminates with a size optimal Delaunay mesh with triangles
having a minimum radius-edge of $\alpha>1$ if 
Lemma~\ref{lemma:diacondition} and Lemma~\ref{lemma:beginastar} are satisfied. 
\end{theorem}

\begin{proof}
First, I prove termination. 
Since Lemma~\ref{lemma:beginastar} is satisfied, the PSLG
is recovered by Delaunay triangulation. 
As $B^*\le A^* + 1/2\log_{e}{2} + 2$ (see 
Lemma~\ref{lemma:allsplit}), it is possible to
compute the minimum value of $A^*$ for which the inequality in 
Lemma~\ref{lemma:diacondition} holds. One has to solve a simple
linear equation to find the minimum value of $A^*$.  
For such an $A^*$, by Lemma~\ref{lemma:aboundsb}, 
it is possible to split the PSLG such that $B^*$ is bounded
as a function of $\alpha$. 

After we split the PSLG segments in a size-optimal manner,
we add Steiner vertices in the mesh such that they are not inside
the diametral circle of a PSLG subsegment. Since no skinny
triangles are formed inside the diametral circle of a PSLG subsegment,
no Steiner vertex encroaches upon a PSLG subsegment. 
Thus, every Steiner vertex will be added inside the domain. 
Since we always place a Steiner vertex such that it is at least a 
distance $\alpha l_0$ from other vertices in the domain, 
the algorithm terminates
when it runs out of space. 

Now, I prove size optimality. 
We saw in Lemma~\ref{lemma:layerk} that the local feature
size at a vertex $p$ of layer $k$ is at most
$l_0(B^* + \alpha + \alpha^2 + ... \alpha^{k-1} + \gamma\alpha^{k-1}),$
where $1<\gamma\le\alpha$.
Due to the off-center vertex insertion algorithm, the length of 
the shortest edge of the skinny triangle 
that results in the insertion of $p$ 
is $l_{\mathrm{short}} = l_0 \frac{\gamma\alpha^{k-1}}{\beta}$. 
Due to the prioritization of skinny triangles with 
shortest edges, the length of the 
shortest edge adjacent to $p$ is at least 
$\alpha l_{\mathrm{short}} = \alpha l_0 \frac{\gamma\alpha^{k-1}}{\beta}$.  
The maximum ratio of the local feature size at $p$
and the length of an edge adjacent to $p$ is
$$\frac{l_0(B^* + \alpha + \alpha^2 + ... \alpha^{k-1} + \gamma\alpha^{k-1})}
{\alpha l_0 \frac{\gamma\alpha^{k-1}}{\beta}} = 
\frac{\beta B^*}{\gamma\alpha^{k}} + \frac{\beta}{\gamma\alpha}
\left(\frac{1}{\alpha^{k-2}} + \frac{1}{\alpha^{k-3}} + ... + 1 \right) + \frac{\beta}{\alpha}.$$
Since $\alpha>1$, $\beta/\alpha < 2$ for $0<\theta^*<\pi/6$, and $\gamma>1$, 
the expression above is less than
$$
2B^* + 2\left(\frac{1}{\alpha^{k-2}} + \frac{1}{\alpha^{k-3}} + ... + 1 \right) + 2 <
2\left(B^* + \frac{\alpha}{\alpha-1} + 1\right).
$$
As the ratio is bounded, my algorithm terminates with a 
size-optimal mesh. 
\end{proof}

\subsubsection{Constrained Delaunay Refinement}
Any mesh that is truly Delaunay is also constrained
Delaunay. Constrained Delaunay meshes, however, may
be smaller than truly Delaunay meshes, so I will derive
the conditions (on $A^*$, $B^*$, and $R$) that are less
strict than conditions derived in the previous section. 
In the previous section, the conditions ensure that no
vertex could be inserted inside the diametral circle of
a PSLG subsegment.  In this section, I derive similar
conditions that ensure that no skinny triangles are 
formed that are inside the diametral circle of a PSLG
segment even as Steiner vertices are added in it. 

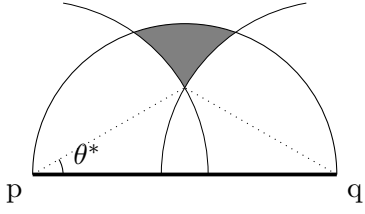
\begin{figure}
\centering
\begin{tikzpicture}[scale=4, invclip/.style={insert path={(-1,-1) rectangle (2,2)}}]

    \draw [line width=0.50mm] (0,0) -- (1,0);
    \draw [dotted] (0,0) -- (0.5,0.288);
    \draw [dotted] (1,0) -- (0.5,0.288);
    \node [below left] at (0,0) {p};
    \node [below right] at (1,0) {q};
    \node [above right] at (0.1,0) {$\theta^*$};
    \draw (0.1,0) arc (0:30:0.1);
    \draw (0,0) arc (180:0:0.5);
    \draw [line width=0.10mm] (0.577,0) arc (0:80:0.577);
    \draw [line width=0.10mm] (0.423,0) arc (180:100:0.577);
    \begin{scope}[on background layer]
     \begin{pgfinterruptboundingbox}
       \clip (0,0) circle (0.577) [invclip];
       \clip (1,0) circle (0.577) [invclip];
       \clip (0.5,0) circle (0.5);
     \end{pgfinterruptboundingbox}
       \fill [color = gray] (0,0) rectangle (1,1);
    \end{scope}
\end{tikzpicture}
\caption{The diametral semicircle of a PSLG subsegment $pq$ is
shown. If the condition in the 
Lemma~\ref{lemma:pslgskinnycondition}
is satisfied, no vertex is placed within the area bounded by
the arcs centered at $p$ and $q$ and the diametral semicircle.
Vertices are allowed only in the shaded region. If $\theta^*>0$, 
the maximum possible distance between two points in the shaded 
region is half the length of subsegment $pq$.
}
\label{fig:cdt_dia}
\end{figure}

\begin{lemma}
\label{lemma:pslgskinnycondition}
Let $\theta^*$ be the desired minimum angle in a mesh
and $\alpha = 1/(2\sin{\theta^*})$ be the desired maximum
radius-edge ratio.  If
$$\frac{B^*}{A^*} + \frac{\alpha}{A^*(\alpha-1)} + \frac{2}{A^*} \le 2\cos{(\theta^*)},$$
no Steiner vertex $a$ will be placed in the diametral 
circle of a PSLG subsegment $pq$ such that 
$\angle apq < \theta^*$ or $\angle aqp < \theta^*$. 
\end{lemma}

\begin{proof}
The proof is nearly identical to the proof of 
Lemma~\ref{lemma:diacondition}.  In the proof, 
the condition in Eq.~\ref{eqn:sqrt2} is derived
to ensure that no vertex is at a distance of
less than $|pq|/\sqrt{2}$ from $p$ or $q$. If we
ensure that no vertex is at a distance less than
$|pq|/2\cos{(\theta^*)}$ from $p$ or $q$, the lemma
holds.  This means that the gray region in 
Fig.~\ref{fig:cdt_dia} is the only region within
the diametral circle of $pq$ where a vertex may 
be placed.  
To prove this lemma, the $\sqrt{2}$ in the 
RHS of Eq.~\ref{eqn:sqrt2} should be
replaced with $2\cos{\theta^*}$. The condition in this
lemma follows the proof of 
Lemma~\ref{lemma:diacondition} from this point
onward. 
\end{proof}

The condition in Lemma~\ref{lemma:pslgskinnycondition} 
above ensures that
there are no skinny triangles adjacent to a
PSLG subsegment. We also need to ensure that
there are no skinny triangles inside
the diametral circle of a PSLG subsegment
that are not adjacent to the subsegment. I 
derive the conditions in the lemmas below.  
First, I will show that if the minimum length of
the shortest segment in a skinny triangle is
bounded from below, the length of the longest
edge in the triangle is also bounded from below. 
I then use this fact and ensure that the
length of any segment in the diametral circle
is also bounded, and thus, skinny triangles inside
the diametral circle of a PSLG subsegment
are impossible. 

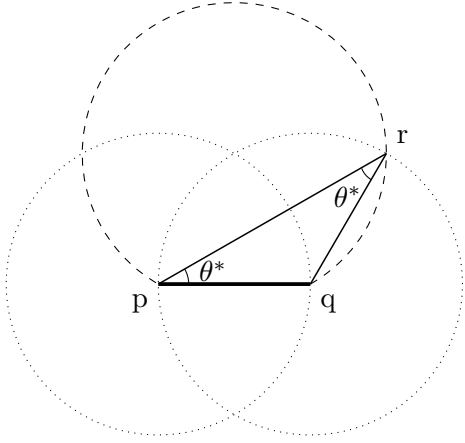
\begin{figure}
\centering
\begin{tikzpicture}[scale=2]
    \draw [line width=0.50mm] (0,0) -- (1,0);
    \draw [line width=0.20mm] (0,0) -- (1.5,0.866);
    \draw [line width=0.20mm] (1,0) -- (1.5,0.866);
    \draw [dotted] (0,0) circle [radius=1];
    \draw [dotted] (1,0) circle [radius=1];
    \draw [dashed] (1,0) arc (-60:240:1);
    \node [below left] at (0,0) {p};
    \node [above right] at (0.2,-0.05) {$\theta^*$};
    \node [below right] at (1,0) {q};
    \node [above right] at (1.5,0.866) {r};
    \draw (0.2,0) arc (0:30:0.2);
    \draw (1.4,0.6928) arc (240:200:0.1414);
    \node [below left] at (1.4,0.715) {$\theta^*$};
\end{tikzpicture}
\caption{The shortest edge of a skinny triangle
is $pq$.  The arc is the locus of points at which
$pq$ subtends an angle $\theta^*$, which is the 
threshold for a triangle to be considered skinny. 
The third vertex should be outside the
dashed arc $pq$. Since $pq$ is the shortest side, the
third vertex should be outside the circles centered
at $p$ and $q$ with the radii equal to the length of
$pq$.  The length of the longest edge is, therefore, 
at least the distance between $p$ (or $q$) and the
point of intersection of the arc and one of the 
circles.}
\label{fig:skinnyarc}
\end{figure}

\begin{lemma}
\label{lemma:skinnylonglength}
If the length of the shortest edge in a skinny
triangle is greater than $l$, the length of the 
longest edge is greater than 
$2l\cos{(\theta^*)}$, 
where $\theta^*$ is the
minimum angle threshold for skinny triangles. 
\end{lemma}

\begin{proof}
Let $pq$ be the shortest edge of a skinny triangle.  
The angle opposite the shortest edge is the smallest
angle, and thus, the third vertex $r$ should be outside
the dashed arc passing through $p$ and $q$ 
in Fig.~\ref{fig:skinnyarc}.  In addition, 
since $pq$ is the shortest side, $r$ should be outside
the circles with centers $p$ and $q$ and radius $l$, 
where $l$ is the length of $pq$.  Clearly, the
shortest possible length of the longest segment
is when $r$ is at the intersection of the circle
and the arc. This length can be calculated
using the cosine rule. Thus, the length of the
longest edge is greater than 
$\sqrt{l^2 + l^2 - 2l^2\cos{(\pi-2\theta^*)}} = 
\sqrt{2l^2 (1 + \cos{2\theta^*})} =
\left(\sqrt{2(1+\cos{(2\theta^*)})}\right)l = 
\left(\sqrt{2(2\cos^2{(\theta^*)})}\right)l =
2l\cos{(\theta^*)}$.
\end{proof}

\begin{lemma}
\label{lemma:diaskinny}
If the condition in Lemma~\ref{lemma:pslgskinnycondition}
is satisfied and 
$$\frac{A^* - \frac{1}{\sqrt{2}}}
{\left( B^* + \frac{\alpha}{\alpha-1} + 1\right)}
> \frac{1}{2\cos{(\theta^*)}},$$
where $\theta^*$ is the minimum desired angle in the mesh,
there will be no skinny triangles formed completely inside 
the diametral circle of a PSLG subsegment
before any vertex encroachment.
\end{lemma}

\begin{proof}
Let the length of the PSLG subsegment under consideration
be $l$. The minimum local feature size at the vertices
of the subsegment is $A^*l$.  The minimum feature size
inside the diametral circle on the PSLG subsegment is
at least $A^*l - (1/\sqrt{2})l$ (because the farthest point
inside the diametral circle away from either vertex is
at a distance of $(1/\sqrt{2})l$).  Due to 
Theorem~\ref{theorem:truesizeoptimal}, the minimum length
of any edge inside the diametral circle (before a vertex 
encroachment) is given by
$$\frac{A^*l - \frac{1}{\sqrt{2}}l}
{2 \left( B^* + \frac{\alpha}{\alpha-1} + 1\right)}.$$
The condition in the 
Lemma~\ref{lemma:pslgskinnycondition} does not allow
Steiner vertices within the circles centered at $p$ and $q$
in Fig.~\ref{fig:cdt_dia}.  Steiner vertices may be present
in the shaded region. It is easy to show that the length of 
the longest edge in the shaded region is less than $l/2$. 
Due to 
Lemma~\ref{lemma:skinnylonglength}, if the length of
the shortest edge in the diametral circle is greater than
$(l/2)/(2\cos{(\theta^*)})$, any possible skinny triangle
will be partly outside the diametral circle of the 
PSLG subsegment. Therefore, if
$$\frac{A^*l - \frac{1}{\sqrt{2}}l}
{2 \left( B^* + \frac{\alpha}{\alpha-1} + 1\right)}
>\frac{l}{4\cos{(\theta^*)}},$$
the lemma is proved.  
Canceling $l/2$ on both sides of the inequality proves
the lemma.
\end{proof}

Note that since $B^* \le A^* + c$, where $c$ is a constant
(see Lemma~\ref{lemma:aboundsb}), the condition in the 
Lemma~\ref{lemma:diaskinny} is linear in $A^*$. Thus, it is
possible to compute the minimum $A^*$ such that the condition is
satisfied. 

\begin{lemma}
If the conditions in Lemmas~\ref{lemma:beginastar},
~\ref{lemma:pslgskinnycondition}, and~\ref{lemma:diaskinny} 
are satisfied, no skinny triangles
will be formed in the diametral circle of a 
PSLG subsegment. Thus, no subsegment is encroached
upon by a Steiner vertex. 
\end{lemma}

\begin{proof}
In the beginning of the execution of the algorithm, 
I ensure that there are no skinny triangles in
the diametral circle of a PSLG subsegment by
enforcing the condition in Lemma~\ref{lemma:beginastar}.
Lemmas~\ref{lemma:pslgskinnycondition} and 
~\ref{lemma:diaskinny} ensure that it is not possible 
to have a skinny triangle inside the diametral 
circle. Thus, we ensure that there will not
be a Steiner vertex that encroaches upon a
PSLG subsegment at any time in the execution
of the algorithm. 
\end{proof}

\begin{theorem}
The algorithm terminates with a size-optimal, constrained Delaunay 
mesh with triangles having a minimum radius-edge of $\alpha>1$ 
if Lemmas~\ref{lemma:beginastar},~\ref{lemma:pslgskinnycondition} 
and~\ref{lemma:diaskinny} above are satisfied. 
\end{theorem}

\begin{proof}
The proof is identical to the proof of 
Theorem~\ref{theorem:truesizeoptimal}.  
As explained above, I have shown that 
if Lemmas~\ref{lemma:beginastar},~\ref{lemma:pslgskinnycondition} 
and~\ref{lemma:diaskinny} above are satisfied, 
there will be no skinny triangles
formed in the diametral circle of a PSLG subsegment.
Thus, there will not be any encroachment. 
Consequently, the constant associated
with the size optimality is identical to the one
obtained in Theorem~\ref{theorem:truesizeoptimal}. The values
of $A^*$ and $B^*$ are possibly smaller because
the conditions they need to satisfy are
less strict.
\end{proof}

%% file: analysis2.tex
\subsection{PSLG with Small Angles}
Thus far in the analysis, we have assumed no angles 
were smaller than $\pi/2$ in the input PSLG.  
The proofs above do not hold when a PSLG angle is small.
because two end points ($p$ and $q$, say)
on adjacent segments may be arbitrarily close to 
each other when the angle between the segments is 
arbitrarily small (see Fig.~\ref{fig:smallangleproblem}).  
Thus, the results above do not hold when the path to the 
ancestral vertex from a Steiner vertex is determined
by such an arbitrarily short edge, which is adjacent
to a vertex on a PSLG segment, whose length is
less than $\mathrm{LFS}(p)/B^*$ or $\mathrm{LFS}(q)/B^*$.  
If the length of the edge is too large (greater than
$\mathrm{LFS}(p)/A^*$ or $\mathrm{LFS}(q)/A^*$, our results
hold because any Steiner vertex inserted due to the edge
simply skips a few layers as do other Steiner vertices
in the mesh.  Pav et al.~\cite{MPW05, P03} decided to ignore
skinny triangles whose shortest edge's end points
lie on two PSLG segments that
meet at a small angle (triangles ``across'' a 
small angle).  My algorithm does the same; it ignores
the triangles, but only if their lengths are smaller than 
$\mathrm{LFS}(p)/B^*$ or $\mathrm{LFS}(q)/B^*$. 
Thus, no path from a Steiner vertex to its 
ancestral vertex is determined by such an arbitrarily short
edge. As a result, the analysis in the subsection above is 
valid for those triangles that are not across a small angle
with a short edge.  In this subsection, I will bound 
the minimum and the maximum angles of such ignored skinny 
triangles as a function of $A^*$, $B^*$, and $R$. 
I begin with a simple lemma below. 

\begin{lemma}
\label{lemma:R}
The ratio of the lengths of adjacent subsegments in the PSLG
is at most $R$.
\end{lemma}

\begin{proof}
Adjacent subsegments share an end point. We 
split all segments such that the lengths of the subsegments
are asymptotically proportional to the LFS at their end points.
The ratio of constants associated with the asymptotic 
proportionality is $B^*/A^* = R$.
\end{proof}

\subsubsection{Truly Delaunay Refinement: PSLG Recovery}
Here, I will show that it is possible to recover the PSLG
by highly refining the input segments in the PSLG and
constructing its Delaunay triangulation. As I have noted
before, if the diametral circle of every subsegment is empty,
the PSLG is recovered. 
First, I will show that there may be only finitely
many subsegments (as a function of $\phi$ and $R$)
whose diametral circle contains some part of an
adjacent PSLG segment. 

\begin{lemma}
\label{lemma:finite}
The number of subsegments whose diametral circles
contain a part of an adjacent segment is bounded
from above as a function of $R$ and $\phi$, where
$\phi$ is the angle between the two segments,
if $R$ is sufficiently small. 
\end{lemma}

\begin{proof}
Let $xp$ and $xq$ be two adjacent segments at
angle $\phi$. Let vertices on $xp$ be placed
at points $p_0$, $p_1$, ..., $p_n$, $p_{n+1}$, and so on. 
Let $|pp_0| = l$.
By Lemma~\ref{lemma:R}, $p_0p_1 = \lambda_0 l$, where 
$1/R \le \lambda_0 \le R$. Similarly,
$|p_kp_{k+1}|=\lambda_0\lambda_1...\lambda_k l$, where
$1/R \le \lambda_i \le R$ for all $0\le i\le n$. 
The distance of the midpoint of $p_np_{n+1}$ from the line segment 
$xq$ is greater than the distance of $p_{n}$ from $xq$, which is
$(1 + \lambda_0 + \lambda_0\lambda_1 + ... + \lambda_0\lambda_1...\lambda_{n-1}) l \sin{\phi}$.
If this distance is greater than the radius $r$ of the diametral circle on $p_np_{n+1}$, 
$r = \frac{1}{2}|p_np_{n+1}| = \frac{1}{2}\lambda_0\lambda_1...\lambda_n l$,
no part of the adjacent segment $xq$ will be inside the diametral circle of 
subsegment $p_np_{n+1}$.  
We should find the  maximum $n$ such that
$$(1 + \lambda_0 + \lambda_0\lambda_1 + ... + \lambda_0\lambda_1...\lambda_{n-1}) l \sin{\phi} \le \frac{1}{2} \lambda_0\lambda_1...\lambda_n l.$$
Canceling $l$ and rearranging the equation above, 
$$\frac{1}{\lambda_0\lambda_1...\lambda_{n-1}} + ... + \frac{1}{\lambda_0} + 1 < \frac{1}{2\sin{\phi}}.$$
Clearly, $n$ is maximized when each of the terms in the LHS
is minimized, which implies that our value of $\lambda_i$ $\forall$ $i$
should be maximized, which is $R$.  In that case,
we have to maximize $n$ for which
$$1 + \frac{1}{R} + \frac{1}{R^2} + ... + \frac{1}{R^{n-1}} < \frac{1}{2\sin{\phi}}.$$
This is geometric progression, and it can be simplified to
$$\frac{R \left(1-\frac{1}{R^{n-1}}\right)}{R-1} < \frac{1}{2\sin{\phi}}.$$
For an empty diametral circle, we do not want the inequality above 
to hold.  If $R$ is large, the inequality always holds for any $n$ no
matter how large $n$ is. Thus, $R$ has to small enough 
(as a function of $\phi$) for the 
inequality to not hold. For a given $R$ that is adequately small,
the inequality only holds when $n$ is adequately small.
If $n$ is any larger, the inequality does not hold, and
the corresponding diametral circles of PSLG subsegments are empty. 
As $R$ tends to 1, the minimum $n$ for which the
inequality does not hold becomes progressively smaller. 
Thus, there are only finitely many subsegments for which
their diametral circle may contain a part of the adjacent segment
when $R$ is sufficiently small. 
\end{proof}

The proof above shows that $R$ has to be sufficiently
small to recover the PSLG segments for truly Delaunay
meshes. The recovery of the PSLG is not an issue for 
constrained Delaunay meshes as the recovery is enforced. 

Second, I will show that when the mesh is highly refined,
no vertex of the adjacent segment will be present inside
the diametral circles of a subsegment. The proof below 
exploits the fact that there are only finitely many 
subsegments whose diametral circle contains part 
of an adjacent segment. 
As $R$ tends to 1, the lengths of subsegments on the PSLG
become shorter, and those finite number of subsegments end up
progressively closer to the vertex with the small angle. Since
the LFS does not vary much near the vertex, those subsegments
tend to have the same length. As a result, the PSLG is recovered
for a sufficiently small $R$. 

\begin{figure}
\centering
\begin{subfigure}[b]{0.45\textwidth}
\begin{tikzpicture}[scale=7]
  \draw [line width=0.25mm] (0,0) -- (.8,0);
  \draw [line width=0.25mm] (0,0) -- (10:.8);
  \draw [fill=black] (0,0) circle (0.01);
  \draw [fill=black] (0:0.2) circle (0.01);
  \draw [fill=black] (0:0.4) circle (0.01);
  \draw [fill=black] (0:0.6) circle (0.01);
  \draw [fill=black] (0:0.8) circle (0.01);
  \draw [fill=black] (10:0.2) circle (0.01);
  \draw [fill=black] (10:0.4) circle (0.01);
  \draw [fill=black] (10:0.6) circle (0.01);
  \draw [fill=black] (10:0.8) circle (0.01);

  \draw (0:0) arc (180:0:0.1);
  \draw (0:0.2) arc (180:0:0.1);
  \draw (0:0.4) arc (180:0:0.1);
  \draw (0:0.6) arc (180:0:0.1);

  \node [below left] at (0,0) {$o$};
  \node [below] at (0.2,0) {$a_1$};
  \node [below right] at (0.8,0) {$a$};
\end{tikzpicture}
\caption{The Limiting Case}
\end{subfigure}
\centering
\begin{subfigure}[b]{0.45\textwidth}
\begin{tikzpicture}[scale=7, invclip/.style={insert path={(-1,-1) rectangle (2,2)}}]
  \draw [fill=black] (0,0) circle (0.01);
  \draw [fill=black] (0:0.2) circle (0.01);
  \draw [fill=gray] (0:0.4) circle (0.02);
  \draw [fill=gray] (0:0.6) circle (0.03);
  \draw [fill=gray] (0:0.8) circle (0.04);
  \draw [fill=gray] (10:0.2) circle (0.02);
  \draw [fill=gray] (10:0.4) circle (0.03);
  \draw [fill=gray] (10:0.6) circle (0.04);
  \draw [fill=gray] (10:0.8) circle (0.05);

  \draw [line width=0.25mm] (0,0) -- (.8,0);
  \draw [line width=0.25mm] (0,0) -- (10:.8);

  \draw (0:0) arc (180:0:0.1);
  \begin{scope}[on background layer]
     \begin{pgfinterruptboundingbox}
        \clip (0.29,0) circle (0.09) [invclip];
     \end{pgfinterruptboundingbox}
     \draw [pattern = north west lines] (0:0.2) arc (180:0:0.11);
  \end{scope}
  \draw (0:0.2) arc (180:0:0.11);
  \draw (0:0.2) arc (180:0:0.09);

  \begin{scope}[on background layer]
     \begin{pgfinterruptboundingbox}
        \clip (0.495,0) circle (0.075) [invclip];
     \end{pgfinterruptboundingbox}
     \draw [pattern = north west lines] (0:0.38) arc (180:0:0.125);
  \end{scope}
  \draw (0:0.38) arc (180:0:0.125);
  \draw (0:0.42) arc (180:0:0.075);

  \begin{scope}[on background layer]
     \begin{pgfinterruptboundingbox}
        \clip (0.695,0) circle (0.065) [invclip];
     \end{pgfinterruptboundingbox}
     \draw [pattern = north west lines] (0:0.57) arc (180:0:0.135);
  \end{scope}
  \draw (0:0.57) arc (180:0:0.135);
  \draw (0:0.63) arc (180:0:0.065);

  \node [below left] at (0,0) {$o$};
  \node [below] at (0.2,0) {$a_1$};
  \node [below right] at (0.8,0) {$a$};
\end{tikzpicture}
\caption{A Generic Case}
\end{subfigure}
\caption{
The diametral circles on PSLG subsegments at a small angle $\phi$. The
length of the subsegment $oa_1$, which is closest to $o$ on the 
horizontal PSLG segment $oa$, has been normalized and set to 1. 
(a) In the hypothetical limiting case when $R=1$, all subsegments 
have an equal length. In this case, the diametral circles do not contain
vertices from the adjacent segment.
(b) When $R>1$, there is a window for each vertex within which 
the vertices have to lie. The windows are shown as filled gray circles, 
but the vertices have to lie on the line segment. 
Consequently, the diametral circles 
have windows, too.  The windows are shown with a pattern of 
diagonal lines. 
When $R$ is sufficiently small, the windows of the vertices
on a segment and the windows of diametral circles of subsegments
on an  adjacent segment do not overlap. 
}
\label{fig:window}
\end{figure}
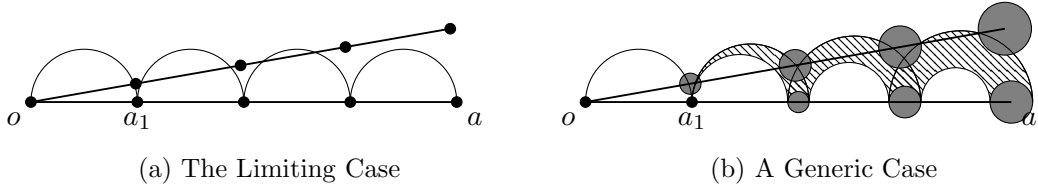

\begin{lemma}
\label{lemma:diawindow}
There exists an $R>1$ below which a diametral circle of a
subsegment does not contain vertices from an adjacent
segment. 
\end{lemma}

\begin{proof}
Hypothetically, when $R=1$, all subsegments in adjacent
segments are of equal length (a consequence of 
Lemma~\ref{lemma:R}). If we increase $R$,
the window inside which a vertex and, therefore, a 
diametral circle lie grows (see Fig.~\ref{fig:window})
because the ratio of lengths of adjacent
subsegments lies between $1/R$ and $R$, so the possible 
locations of every vertex grow with $R$. 
In the proof of Lemma~\ref{lemma:finite} above,
by setting all $\lambda_i$ as $1/R$ and then $R$, it
is possible to compute the range of the window for every
vertex on the PSLG line segment. 
We saw in Lemma~\ref{lemma:finite} that we have to consider
only finitely many vertices (as a function of $\phi$).  
As we reduce $R$, there will be some value at which
the windows do not intersect.  
Below that $R$, the diametral circles are empty. 
Note the length
of $oa_1$ in the diagrams in Fig.~\ref{fig:window} 
has been normalized and set to $1$ (because this length
changes as $R$ changes), but the argument in this 
proof still holds. 
\end{proof}

\subsubsection{Truly and Constrained Delaunay Refinement:\\ Minimum and Maximum Angles}
I will first show the bounds on the minimum and the maximum
angle in triangles across a small angle. These bounds are
applicable for both truly and constrained Delaunay meshes.
First, I will show the bounds for the minimum angle.  The proof
presented here is almost identical to the one presented
by Pav et al.~\cite{MPW05,P03} in their proof for the bound on the
minimum angle.  

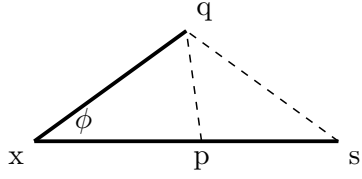
\begin{figure}
\centering
\begin{tikzpicture}[scale=4]
    \draw [line width=0.50mm] (0,0) -- (1,0);
    \draw [line width=0.50mm] (0,0) -- (0.5,0.366);
    \draw [line width=0.20mm, dashed] (1,0) -- (0.5,0.366);
    \draw [line width=0.20mm, dashed] (0.55,0) -- (0.5,0.366);
    \node [below left] at (0,0) {x};
    \node [below] at (0.55,0) {p};
    \node [below right] at (1,0) {s};
    \node [above right] at (0.5,0.366) {q};
    \node [above right] at (0.10,-0.01) {$\phi$};
\end{tikzpicture}
\caption{
A part of a PSLG is shown as thick lines. The thin,
dashed lines are for reference. We assume
$|xp|\le|xq|$ and that all segments in the PSLG 
have been split using at least one vertex.
}
\label{fig:smallangle}
\end{figure}

\begin{theorem}
\label{theorem:minangle}
As the PSLG is progressively refined, the minimum
angle in the resulting meshes from the Delaunay 
refinement of progressively refined segments tends to 
$\arctan{\left(\frac{sin{\phi}}{2-\cos{\phi}}\right)}$,
where $\phi$ is the minimum angle in the PSLG. 
\end{theorem}

\begin{proof}
Consider a part of the split PSLG shown in Fig.~\ref{fig:smallangle}.  
The subsegments $xp$, $ps$, and $xq$ are part of PSLG with a
small angle $\phi$ at $x$. Note that we have to split
all PSLG segments with at least one vertex each 
for the lemma to hold.  Without loss of generality, 
let us assume that $|xp|\le|xq|$. 
By sine rule,
$$\frac{|xs|}{\sin{(\angle xqs)}} = \frac{|xq|}{\sin{(\angle psq)}}.$$ 
Thus, $\sin{(\angle psq)} = \frac{|xq|sin{(\angle xqs)}}{|xs|}$, 
which implies $\sin{(\angle psq)} = \frac{|xq|}{|xp|+|ps|}\sin{(\angle xqs)}$. 
The minimum value of $|xq|$ is $|xp|$ (our assumption), and the maximum value
of $|ps|$ is $R|xp|$ (see Lemma~\ref{lemma:R}).
Therefore,
$$\sin{(\angle psq)} \ge \frac{R|xp|}{(R+1)|xp|}\sin{(\angle xqs)}.$$
Now, $\sin{(\angle xqs)} = \sin{(\pi - \angle qxs- \angle qsx)} = 
\sin{(\angle qxs + \angle qsx)}$.  Therefore, assuming $\angle psq$ 
is acute,
\begin{align*}
\sin{(\angle psq)} &\ge \left(\frac{R}{1+R}\right)\sin{(\angle qxs + \angle qsx)}\\
\implies \left(\frac{1+R}{R}\right)\sin{(\angle psq)} &\ge \sin{(\angle qxs + \angle qsx)}\\
\implies\left(\frac{1+R}{R}\right)\sin{(\angle psq)} &\ge \sin{(\angle qxs)} \cos{(\angle qsx)} + \cos{(\angle qxs)} \sin{(\angle qsx)}\\
\implies\sin{(\angle psq)} &\ge \frac{\sin{(\angle qxs)}\cos{(\angle qsx)}}{\left(\frac{1+R}{R}\right) - \cos{(\angle qxs)}}\\
\implies\tan{(\angle psq)} &\ge \frac{\sin{(\angle qxs)}}{\left(\frac{1+R}{R}\right) - \cos{(\angle qxs)}}\\
\implies\tan{(\angle psq)} &\ge \frac{\sin{(\phi)}}{\left(\frac{1+R}{R}\right) - \cos{(\phi)}}. 
\end{align*}
Let us assume $pq$ is the shortest side of a skinny triangle.
In my Delaunay refinement algorithm, I ignore skinny triangles
on the edge $pq$ since the edge is across a small angle.  
Since the Delaunay triangulation algorithm lexicographically
maximizes the angles among all possible triangulations of
given a set of vertices, any triangle formed on $pq$ on the 
same side as $s$ will have a minimum angle greater than or 
equal to $\angle psq$. As the PSLG is progressively refined, 
the value of $R$ tends to 1.  Thus, the lemma holds. 
\end{proof}

My bounds are slightly weaker than those of Pav et 
al.'s~\cite{MPW05,P03} in that my bounds approach
their bounds only in the limit. It is possible 
to get the same bound as theirs by simply splitting 
all PSLG subsegments into two equal parts after the
ODE-based splits. As the analysis of this slightly
changed algorithm does not provide any additional
insight into Delaunay mesh refinement, I have
focused only on the original algorithm in 
this paper. 

In the next theorem below, I provide the
bounds for the maximum angle.  In 
Lemma~\ref{lemma:diawindow}, I used the fact
that all subsegments near a vertex with 
a small angle tend to have the same length as 
$R$ tends to 1. While that is true, in the
proof below, I will argue only that
pairs of subsegments on adjacent PSLG segments
tend to have the same length. Those pairs
of subsegments are the two subsegments adjacent
to the vertex with the small angle, the two subsegments
next to them, and so on. The reason for presenting 
two slightly different ideas is to provide 
readers with additional insights. 

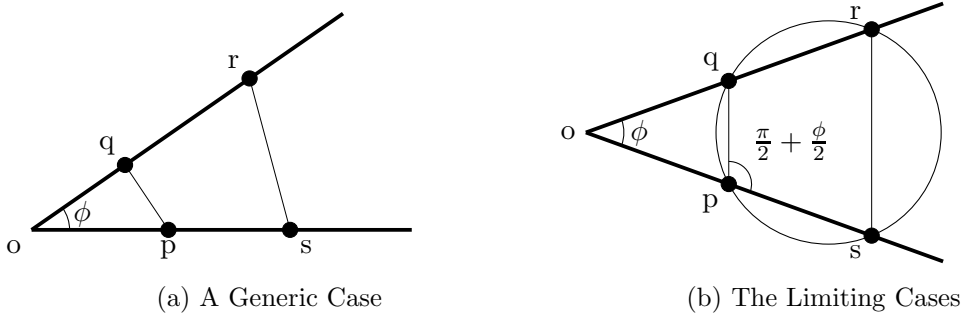
\begin{figure}
\centering
\begin{subfigure}[b]{0.45\textwidth}
\begin{tikzpicture}[scale=1]
    \draw [line width = 0.5mm] (0,0) -- (5,0);
    \draw [line width = 0.5mm] (0,0) -- (35:5);
    \draw [line width = 0.1mm] (0:1.8) -- (35:1.5);
    \draw [line width = 0.1mm] (0:3.4) -- (35:3.5);
    \draw [fill=black] (0:1.8) circle (1mm);
    \draw [fill=black] (0:3.4) circle (1mm);
    \draw [fill=black] (35:1.5) circle (1mm);
    \draw [fill=black] (35:3.5) circle (1mm);
    \draw (0:0) ++(0:0.5) arc (0:35:0.5);

    \node [below left] at (0,0) {o};
    \node [below] at (0:1.8) {p};
    \node [below right] at (0:3.4) {s};
    \node [above left] at (35:1.5) {q};
    \node [above left] at (35:3.5) {r};
    \node at (17.5:0.7) {$\phi$};
\end{tikzpicture}
\caption{A Generic Case}
\end{subfigure}
\centering
\begin{subfigure}[b]{0.45\textwidth}
\begin{tikzpicture}[scale=1]
    \draw [line width = 0.5mm] (0,0) -- (20:5);
    \draw [line width = 0.5mm] (0,0) -- (-20:5);
    \draw [line width = 0.1mm] (20:2) -- (-20:2);
    \draw [line width = 0.1mm] (20:4) -- (-20:4);
    \draw (0:3.1925) circle (1.4806);

    \draw [fill=black] (-20:2) circle (1mm);
    \draw [fill=black] (-20:4) circle (1mm);
    \draw [fill=black] (20:2)  circle (1mm);
    \draw [fill=black] (20:4)  circle (1mm);

    \node [left] at (0,0) {o};
    \node [below left] at (-20:2) {p};
    \node [below left] at (-20:4) {s};
    \node [above left] at (20:2) {q};
    \node [above left] at (20:4) {r};
    \draw (0:0) ++(-20:0.5) arc (-20:20:0.5);
    \draw (-20:2) ++(90:0.3) arc (90:-20:0.3);
    \node at (0:0.7) {$\phi$};
    \node at ($(-20:2)+(35:1)$) {$\frac{\pi}{2} + \frac{\phi}{2}$};
\end{tikzpicture}
\caption{The Limiting Cases}
\end{subfigure}
\caption{
The bound on the maximum angle in the triangle 
across a small angle. The thick lines segment are
part of the input PSLG. (a) This is a generic case that
will be seen when the PSLG is not highly refined.
Notice that $|pq|$ is less than any segment 
adjacent to $p$ or $q$. 
(b) When the PSLG is highly refined, the end points
of mesh edges tend to vertex positions as shown, i.e., 
they will move toward positions such that
$|ps|=|qr|$. The bound on the angle is 
$\pi/2+\phi/2$. 
In contrast to the other diagrams in this paper, 
no PSLG segment is horizontal in this diagram
because there is symmetry along the horizontal
line, which is easy to observe when the PSLG
segments are rotated.  
}
\label{fig:maxangle}
\end{figure}

\begin{theorem}
As the PSLG is progressively refined, the maximum
angle in the resulting meshes from the Delaunay 
refinement of progressively refined segments tends to 
$\pi/2 + \phi/2$,
where $\phi$ is the minimum angle in the PSLG. 
\end{theorem}

\begin{proof}
Consider Fig.~\ref{fig:maxangle}(a), in which
two segments of a PSLG, $os$ and $or$, meet at a 
small angle $\phi$. Let $pq$ be the shortest
edge of a skinny triangle whose length is
smaller than the threshold.  Since the triangle
on $pq$ is not refined, the maximum angle is bounded 
by $\angle spq$, $\angle rqp$, $\angle xpq$, or 
$\angle xqp$.  Even if $pq$ forms a triangle with
some other vertex, its maximum angle will be smaller 
than one of the four angles. Let us consider only 
$\angle spq$ and $\angle rqp$ for now.  Without
loss of generality, consider $\angle spq$. 
Let the length of $ps$ 
be $l$.  The LFS at $p$ is at most $B^*l$. 
If $\triangle spq$ is the skinny triangle that
is ignored by my algorithm, 
$|pq| < \mathrm{LFS}(p)/B^* \le |ps| = l$.  
Since the LFS is Lipschitz function, 
$\mathrm{LFS}(q) \le \mathrm{LFS}(p) + |pq|$, which implies
$\mathrm{LFS}(q) \le B^*l + |pq|$ (because $\mathrm{LFS}(p)\le B^*l $), which implies
$\mathrm{LFS}(q) \le B^*l + l$ (because $|pq|<l $). 
As $|qr|$  is asymptotically size optimal, 
$A^*|qr| \le \mathrm{LFS}(q) \le B^*l + l$, which implies
$|qr| \le \frac{B^*}{A^*}l + \frac{l}{A^*}$.  We can
similarly prove that 
$|qr| \ge \frac{A^*}{B^*}l - \frac{l}{B^*}$ by applying
$\mathrm{LFS}(q) \ge \mathrm{LFS}(p) - |pq|$.  The two
results above can be simplified to
$$\left(R - \frac{1}{B^*}\right)l \le |qr| \le \left(R + \frac{1}{A^*}\right)l.$$
As we refine the PSLG, although the value of $l$ and
the location of the vertices change,
$R$ tends to 1 and $A^*$ and $B^*$ tends to infinity, 
which implies that the lengths of $ps$ and $qr$ both
tend to the same value. Since the pairs of subsegments
tend to have the same length as $R$ tends to 1,
we get the limiting case in Fig.~\ref{fig:maxangle}(b),
where $p$ and $q$ ($r$ and $s$, also) are equidistant 
from $o$. Thus, the maximum angle tends to $\pi/2 + \phi/2$.  
Also note that $\angle xqp$ and $\angle xpq$ are acute in the 
limiting case, and so their magnitude is below the bound.  
\end{proof}

In the appendix, I have shown that a small angle may be
defined as any angle $\phi<\arccos{\frac{1}{2R}}$.  As
we progressively refine the mesh, $R$ tends to $1$, the
threshold on the small angle tends to $\pi/3$, so
the maximum angle tends to 
$\pi/2+\phi/2 = \pi/2 + \frac{\pi}{3}/2 = 2\pi/3$.

Summarizing the results, we have shown the following 
for truly and constrained Delaunay mesh refinement:
As we progressively refine the PSLG segments and then
refine the Delaunay triangulation of the refined 
segments,
\begin{enumerate}
\item The minimum angle in the mesh may be improved to 
$\theta^*<\arctan{\left(\frac{sin{\phi}}{2-\cos{\phi}}\right)}$, where
$\phi$ is the smallest angle in the PSLG. 
\item The minimum angle in triangles that are not 
across a small angle may be improved to $\theta^*<\pi/6$.
\item The maximum angle in the mesh may be less than
some angle strictly greater than $2\pi/3$.  
\end{enumerate}
The meshes are size optimal because the lengths
of edges depends only on the feature size or
magnitude of the angles in the PSLG.  For the same
quality of the mesh, the truly
Delaunay meshes have more vertices, edges, and elements
than constrained Delaunay meshes.

%% file: future.tex
\section{Discussion}
This paper improves upon the results by
Shewchuk~\cite{S97} and Pav~\cite{P03}
mainly through the use of the
ODE-based algorithm, which distributes vertices
on the PSLG such that the subsegments are
size optimal.  The lemmas, theorems, and 
their proofs are algebraic consequences
of the algorithm. In many proofs, 
the ratio $R$ of the upper and lower bound 
on the constant associated with the size optimality 
of split subsegments plays a major role. 
In Chew's algorithm, as the PSLG segments are
split such that their lengths are nearly
identical no matter how large the LFS is, 
the value of $A^*$ and $B^*$ 
in the algorithm might be very high. 
My ODE-based algorithm provides a way
to modify $A^*$ and $B^*$ by simply increasing or 
decreasing the number of splits in the PSLG segment with
the shortest reference segment. The rest
of the algorithms in this paper are identical 
to the ones discussed in Section 2.  All of them
split PSLG segments on-the-fly, which 
does not result in meshes of high quality as
my algorithm yields. 

I use the off-center Steiner vertex insertion 
algorithm with shortest edge prioritization
because it yields smaller meshes in
a practical implementation. This may not
be necessary. My analysis might be
generalized (as in~\cite{CC09,FCC10})
to Steiner vertex insertion at any point 
within the circumcircle such that it is at a 
certain minimum distance (normalized with respect
to the shortest edge of the skinny triangle)
from all the existing vertices in the mesh, but
more research is needed. In addition, a time complexity 
analysis of the algorithm should also be carried out
as it has been done for prior algorithms by
Miller~\cite{M04} and Har-peled and 
\"{U}ng\"{o}r~\cite{HU05}. 

There are several directions for future research, but
none of those directions are straightforward.  
My algorithm may be extended to 3D (or higher) as
Chew's algorithm was extended by Dey, Bajaj, and
Sugihara~\cite{DBS91}. Such an algorithm
will involve solving the problem of constructing
2D meshes on facets in a piecewise
linear complex (PLC) such that the lengths of
the edges are asymptotically proportional
to the LFS at their end points. In my algorithm,
this problem was easily solved using the solution
of the ODE. It is challenging to solve this on
a 2D facet of any shape. 

Another area of research is the generation of
high-order meshes. Currently, only heuristic
algorithms exist to obtain high-order meshes. 
Typically, only the first few layers (from the
boundary) of a high-order mesh are curved.  Since
my algorithm constructs a mesh in an advancing-front 
fashion, it can, perhaps, be extended to construct
high-order 2D meshes with guaranteed quality. 

There are quadrilateral mesh generators, known
as Q-Morph~\cite{OSCS99}, which use a 
triangular mesh as a point-location data structure
to construct a quadrilateral mesh in an
advancing-front fashion. They also use the 
triangular mesh as a guide to estimate the 
size of quadrilateral elements in all parts of
the domain. My algorithm could provide a stepping
stone towards guaranteed-quality advancing-front
quadrilateral mesh generation.

%% file: appendix.tex
\section{Defining Small Angle}
The following lemma provides the justification 
for defining a small angle as any angle 
less than $\arccos{(1/(2R))}$.

\begin{lemma}
\label{lemma:rminangle}
If the angle between two adjacent segments with a
common vertex $o$ is
$\phi>\arccos{(1/(2R))}$, the distance from a vertex $a$
on one of the segments to a vertex on the adjacent
segments is greater than the length of one of the
subsegments adjacent to $a$ that is closer to $o$.
\end{lemma}

\begin{proof}
Consider vertex $o$, which is an end point of two
line segments in the PSLG (see Fig.~\ref{fig:adjacent}(a)). 
Let the segments be $oa$ and $ob$.  Let the vertices
closest to $o$ on $oa$ and $ob$ be $a_1$ and $b_1$,
respectively.  The ratio $oa_1/ob_1$ is greater than 
or equal to $1/R$ and less than or equal to $R$. 
If $|a_1b_1| \ge |oa_1|$,
then $|ob_1| \ge 2|oa_1|\cos{\phi}$ (and vice versa). 
Thus, if $\cos{\phi} \le |ob_1|/2|oa_1| \le 1/(2R)$ (because
$|ob_1|/|oa_1|$ is at least $1/R$), $|a_1b_1| \ge |oa_1|$.
In addition, other vertices on $ob$ are at a greater
distance from $a_1$ than is $b_1$. Thus, we have
proved the lemma for the closest vertex $a_1$ (the same
argument holds for $b_1$ as well). 

I will now consider vertices on $oa$ and $ob$ that are
not the closest to $o$ on their respective segments. 
Consider the vertex $a_2$ that is closest to $a_1$
on the line segment $a_1a$ (see Fig.~\ref{fig:adjacent}(b)). 
Because $a_1a_2$ and
$oa_1$ have a common vertex $a_1$, 
$\frac{|a_1a_2|}{|oa_1|} \le R$.  
The point on $ob$ that is closest to $a_2$ is
the perpendicular projection of $a_2$ on $ob$
at, say, $m$.  The point $m$ is at a distance 
$oa_2\sin{\phi}$ from $a_2$.   
If this distance is less than $|a_1a_2|$, the lemma 
is proved for vertex $a_2$ (sufficient condition; not
a necessary condition).  Let us find the location
of $a_2$ where $|a_1a_2| = |a_2m|$.
Let $|a_1a_2| = k|oa_1|$.
Clearly,
$$\sin{\phi} = \frac{|a_2m|}{|a_1a_2|+|oa_1|} = 
\frac{|a_1a_2|}{|a_1a_2|+|oa_1|} = 
\frac{k|oa_1|}{k|oa_1|+|oa_1|} = \frac{k}{1+k},$$
which implies $k=\frac{1}{1-\sin{\phi}}$, but
$k>R=\frac{1}{2\cos{\phi}}$ for $0<\phi<\pi/2$.
Therefore, $a_2$ will not be placed such that 
its distance from any vertex on $ob$ is less 
than $a_1a_2$. 

The argument made for $a_2$ can be extended to $a_3$,
$a_4$, and so on. Note that 
$\frac{|a_ia_{i+1}|}{|a_{i-1}a_i|} \le R \Rightarrow
\frac{|a_ia_{i+1}|}{|oa_i|} \le R$. From here, the 
argument made in the previous paragraph holds for
all vertices on the PSLG segment $oa$.
\end{proof}

Lemma~\ref{lemma:rminangle} shows that a segment that joins
a vertex $a$ to vertices on adjacent segments has a length 
greater than $\mathrm{LFS}(p)/B^*$.  Thus, $\phi$ is not 
a small angle even if $\phi<\pi/2$. 
Note that $R$ is strictly greater than 1, so any angle
$\phi\le\pi/3$ is a small angle. 

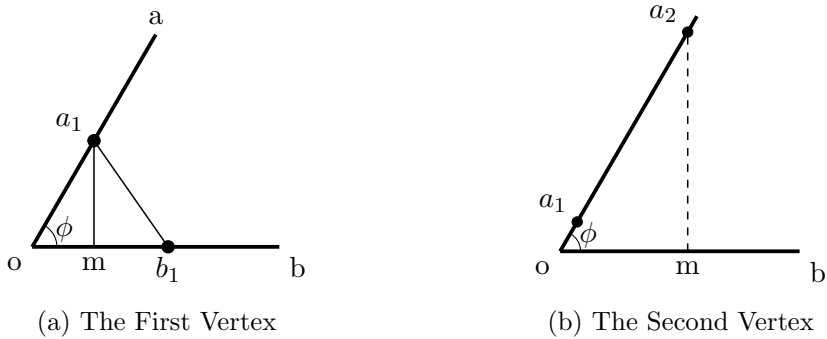
\begin{figure}[b]
\begin{subfigure}[b]{0.45\textwidth}
\centering
\begin{tikzpicture} [scale=3.25]
    \draw [line width=0.50mm] (0,0) -- (1,0);
    \draw [line width=0.50mm] (0:0) -- (60:1);
    \draw [line width=0.20mm] (0.55,0) -- (60:0.5);
    \draw [line width=0.20mm] (0.25,0) -- (60:0.5);
    \draw [fill = black] (0.55,0) circle (0.025);
    \draw [fill = black] (60:0.5) circle (0.025);
    \node [below left] at (0,0) {o}; 
    \node [below right] at (1,0) {b}; 
    \node [above] at (60:1) {a}; 
    \node [above left] at (60:0.5) {$a_1$}; 
    \node [below] at (0.55,0) {$b_1$} ;
    \node [below] at (0.25,0) {m} ;
    \node at (30:0.15) {$\phi$} ;
    \draw  (0:0) ++(0:0.1) arc (0:60:0.1);
\end{tikzpicture}
\caption{The First Vertex}
\end{subfigure}
\begin{subfigure}[b]{0.40\textwidth}
\centering
\begin{tikzpicture} [scale=0.45]
    \draw [line width=0.50mm] (0,0) -- (7,0);
    \draw [line width=0.50mm] (0:0) -- (60:8);
    \draw [fill = black] (60:1) circle (0.15);
    \draw [fill = black] (60:7.46) circle (0.15);
    \draw [line width=0.2mm, dashed] (3.73,0) -- (60:7.46);
    \node [below left] at (0,0) {o}; 
    \node [below right] at (7,0) {b}; 
    \node [below] at (3.73,0) {m}; 
    \node [above left] at (60:1) {$a_1$}; 
    \node [above left] at (60:7.46) {$a_2$}; 
    \node at (30:0.95) {$\phi$} ;
    \draw  (0:0) ++(0:0.6) arc (0:60:0.6);
\end{tikzpicture}
\caption{The Second Vertex}
\end{subfigure}
\caption{
How the value of $R$ affects the lengths of segments from
a vertex on a PSLG segment to a vertex on its adjacent segment. 
(a) The length of $a_1b_1$ must be greater than the length
of $oa_1$, so $ob_1 > 2|a_1b_1|\cos{\phi}$, which happens
when $\cos{\phi}<1/(2R)$. 
(b) The length of $a_2m$ should be greater than $a_1a_2$.
If the ratio of lengths of $oa_1$ and $a_1a_2$ is small enough,
this condition is easily satisfied. 
}
\label{fig:adjacent}
\end{figure}